\documentclass[10pt]{article}

\usepackage{amsmath,amsfonts,amssymb,amsthm,graphicx}
\usepackage[authoryear]{natbib}
\usepackage{etoolbox}
\usepackage{color}
\usepackage[colorlinks=true,citecolor=blue,pdfpagemode=UseNone,pdfstartview=FitH]{hyperref}

\allowdisplaybreaks[4]

\renewcommand{\d}{\,\mathrm{d}}
\newcommand{\dd}{\mathrm{d}}
           
\newcommand{\E}{\mathbb{E}}
\renewcommand{\P}{\mathbb{P}}
\newcommand{\UE}{\overline{\mathbb{E}}}
\newcommand{\UP}{\overline{\mathbb{P}}}
\newcommand{\R}{\mathbb{R}}

\newcommand{\K}{\mathcal{K}}
\newcommand{\Loss}{\mathrm{Loss}}

\newcommand{\FFF}{\mathcal{F}}

\theoremstyle{plain}
\newtheorem{theorem}{Theorem}[section]
\newtheorem{corollary}[theorem]{Corollary}
\newtheorem{lemma}[theorem]{Lemma}
\newtheorem{proposition}[theorem]{Proposition}

\theoremstyle{definition}
\newtheorem{protocol}[theorem]{Protocol}

\theoremstyle{remark}
\newtheorem{remark}[theorem]{Remark}

\newlength{\IndentI}
\newlength{\IndentII}
\newlength{\IndentIII}
\newlength{\IndentIV}
\newlength{\IndentV}
\newcommand{\indentI}{\noindent\hspace*{\IndentI}}
\newcommand{\indentII}{\noindent\hspace*{\IndentII}}
\newcommand{\indentIII}{\noindent\hspace*{\IndentIII}}
\newcommand{\indentIV}{\noindent\hspace*{\IndentIV}}
\newcommand{\indentV}{\noindent\hspace*{\IndentV}}

\newcommand{\abstr}{%
  It is well known that a Bayesian probability forecast for all future observations
  should be a probability measure in order to satisfy a natural condition of coherence.
  The main topics of this paper are the evolution of the Bayesian probability measure
  and ways of testing its adequacy as it evolves over time.
  The process of testing evolving Bayesian beliefs is modelled in terms of betting,
  similarly to the standard Dutch book treatment of coherence.
  The resulting picture is adapted to forecasting several steps ahead
  and making almost optimal decisions.}

\title{The diachronic Bayesian}
\author{Vladimir Vovk}

\begin{document}
\maketitle
\begin{abstract}
  \smallskip
  \abstr

  The version of this paper at \url{http://probabilityandfinance.com} (Working Paper 64)
  is updated most often.
\end{abstract}

\section{Introduction}

Consider a Bayesian forecaster predicting future observations.
Two standard examples, which can be considered as the opposite ends of a spectrum,
are where the observations are outcomes of coin tosses (for a possibly biased coin)
and where the observations are ``dry'' or ``rain'' for a number of consecutive days.
Let us take the standard Bayesian position,
due to \citet{deFinetti:1937,deFinetti:2017}
and discussed by, e.g., \citet[Sect.~4.1]{Bernardo/Smith:2000},
that the Bayesian's beliefs about the future observations should be encoded as a probability measure
on the sequences of observations.

A fundamental role in de Finetti's theory is played by the requirement of \emph{coherence}:
if the Bayesian's beliefs do not form a probability measure,
we can set a ``Dutch book'' against him,
which is a system of bets leading to his sure loss.
We will also be interested in a stronger property,
agreement with reality,
in which sure loss is replaced by a substantial gain for an opponent
who follows a strategy that can ever lose only a tiny amount.
Both coherence and agreement with reality are defined in terms of betting.

Therefore, we assume that at each point in time
the Bayesian has a probability measure representing his beliefs for the future observations.
In this paper we are interested in how the Bayesian's probability measure changes over time.
A standard simple answer is that we should include in our prediction picture
\textbf{all} information that the Bayesian gets,
and then we should condition on the new information
in the usual sense of probability theory \citep[Sect.~I.4]{Kolmogorov:1933-Latin}.
This procedure for updating the Bayesian's beliefs is known as ``Bayesian conditioning''
\citep{Bayes:1764,Shafer:1982,Shafer:2022event}.
The principle that the new observations  must be the only thing
the Bayesian has learned is the \emph{principle of total evidence} \citep{Shafer:1985},
and it is often regarded as uncontroversial.
\citet{Lewis:1999} derives Bayesian conditioning (as updating rule)
via his diachronic Dutch book result,
which implicitly relies on the principle of total evidence.
In Sect.~\ref{sec:basic} we will discuss the narrowness of the principle of total evidence
and, therefore, of Bayesian conditioning.
While it may be convincing in coin-type situations, it is not in weather-type ones
(cf.\ the first paragraph of this section).
The main mathematical observation of Sect.~\ref{sec:basic}
is that Lewis's requirement of no diachronic Dutch book
does not impose any restrictions on the forecasts at different times,
so for discussing the diachronic aspects of Bayesian forecasting
we need a stronger requirement.

Section~\ref{sec:test} proposes a testing protocol
based on betting for the evolving Bayesian probability measure.
This protocol is given in terms of observables and does not depend on Bayesian conditioning.
It is very much in the spirit of the work on game-theoretic probability
\citep{Shafer/Vovk:2019,Dawid/Vovk:1999,Vovk:1993}
and the recent RSS discussion papers
by \citet{Shafer:2021}, \citet{Waudby-Smith/Ramdas:2024}, and \citet{Grunwald/etal:2024}
promoting game-theoretic statistics.

A discussion of connections
with the standard measure-theoretic picture of probability and statistics
follows in Sect.~\ref{sec:embedding}.
The measure-theoretic picture will typically be an imaginary picture
that does involve Bayesian conditioning
(which may be happening deeply inside the imaginary picture, far from what we can observe).
The main finding of this section
is the equivalence of the game-theoretic and measure-theoretic pictures
for finite probability spaces.

Section~\ref{sec:K} adapts the testing protocol of Sect.~\ref{sec:test}
to predicting $K$ steps ahead,
which generalizes the case $K=1$ considered earlier in game-theoretic probability
(e.g., \citealp{Shafer/Vovk:2019,Dawid/Vovk:1999,Vovk:1993}).
Section~\ref{sec:DM} applies the $K$-steps-ahead testing protocol
to making nearly optimal decisions,
and Sect.~\ref{sec:conclusion} concludes.

Appendixes~\ref{app:proofs}--\ref{app:martingale-SLLN}
provide further information.
The key one is Appendix~\ref{app:proofs} giving the proofs.
Appendix~\ref{app:ideal} discusses a seemingly more general (but in fact equivalent)
testing procedure,
and Appendix~\ref{app:radical} discusses Jeffrey's radical probabilism
in the context of testing diachronic Bayesian predictions.

This paper has been inspired by Philip Dawid's brief discussion of one-step-ahead prediction
\citep[Sect.~7]{Vovk/Shafer:2023},
and its title is adapted from \citet{Dawid:1982Bayesian}
(being well-calibrated is an important aspect, namely the frequency aspect, of agreement with reality).
Its other important source is Dawid's super-strong prequential principle
\citep[Sect.~5.2]{Dawid/Vovk:1999}.
This principle requires that our testing protocol based on betting
must agree with the measure-theoretic picture, regardless of the imagined data-generating distribution.

\subsection{Predictivism}

A very common Bayesian picture is the more complicated one where,
instead of one probability measure $P$ over the observations,
the Bayesian's beliefs are modelled as a statistical model
$\{P_{\theta}\mid\theta\in\Theta\}$
combined with a prior probability measure $\mu$ on $\Theta$.
A standard Bayesian point of view \citep[Chap.~4]{Bernardo/Smith:2000}
is that we should start from one probability measure $P$ and then,
if this is more convenient, e.g., mathematically,
represent it as integral $P=\int P_{\theta}\mu(\dd\theta)$.
An example is the application of de Finetti's theorem to coin tossing,
guaranteeing that any exchangeable probability measure $P$
can be represented as a mixture $\int P_{\theta}\mu(\dd\theta)$
of probability measures $P_{\theta}$ corresponding
to independent and identically distributed observations.
In Lindley's words,
``We should be concentrating not on Greek letters but on the Roman letters''
(i.e., not on $\theta$s, parameter values, but on the $x$s and $y$s, observables)
\citep[Sect.~7]{Vovk/Shafer:2023}.
This view is sometimes called predictivism \citep{Wechsler:1993}.

\subsection{Notation and terminology}

If $a$ and $b$ are finite sequences (of some elements),
we write $a\subseteq b$ to mean that $a$ is a prefix of $b$,
and we write $a\subset b$ to mean that $a\subseteq b$ and $a\ne b$.
If $a\subseteq b$,
$b\setminus a$ is the sequence obtained from $b$ by crossing out its prefix $a$.
The concatenation of $a$ and $b$ is written simply as $ab$;
we use the same notation when $a$ or $b$ (or both) are elements;
if $B$ is a set of elements or finite sequences, $aB$ stands for $\{ab\mid b\in B\}$.
The length of a finite sequence $a$ is denoted by $\lvert a\rvert$.

If $a$ and $b$ are numbers, $a\wedge b:=\min(a,b)$.

In our terminology we will mainly follow \citet{Shiryaev:2016,Shiryaev:2019}.
A \emph{finite probability space} is a pair $(\Omega,P)$,
where $\Omega$ is a finite set,
implicitly equipped with the $\sigma$-algebra $\FFF$ of all subsets of $\Omega$,
and $P$ is a probability measure on $(\Omega,\FFF)$.
Let us say that $(\Omega,P)$ is \emph{positive}
if the probability of each sample point is positive,
$P(\{\omega\})>0$ for all $\omega\in\Omega$.
We let $\E_P$ stand for the expected value under $P$.

We will also use the following notation:
\begin{itemize}
\item
  $\mathfrak{P}(A)$ is the set of all positive probability measures
  on a finite set $A$;
\item
  $\R^A$ is the set of all functions $f:A\to\R$,
  $\R$ being the real line;
\item
  if $P\in\mathfrak{P}(\mathbf{Y}^K)$ and $x\in\mathbf{Y}^k$ for $k\le K$,
  \[
    P(x)
    :=
    P(x\mathbf{Y}^{K-k});
  \]
  in this paper we use this notation, and $P(x'\mid x)$ introduced next,
  for a finite~$\mathbf{Y}$;
\item
  if $P\in\mathfrak{P}(\mathbf{Y}^K)$, $x\in\mathbf{Y}^k$, and $x'\in\mathbf{Y}^{k'}$
  for $k+k'\le K$,
  \[
    P(x'\mid x)
    :=
    \frac{P(x x')}{P(x)}
  \]
  (under our definitions the denominator will always be positive).
\end{itemize}

A \emph{filtration} $(\FFF_n)$ in a finite probability space $(\Omega,P)$,
where $n$ ranges over a contiguous set of integers, 
is an increasing sequence of $\sigma$-algebras in $\Omega$,
$\FFF_{n_1}\subseteq\FFF_{n_2}$ when $n_1\le n_2$.
We say that a sequence $(Y_n)$ of random variables in $(\Omega,P)$ is \emph{adapted}
if $Y_n$ is $\FFF_n$-measurable for all $n$;
it is \emph{predictable} if $Y_n$ is $\FFF_{n-1}$-measurable for all $n$.

In this paper we mainly concentrate on finite probability spaces.
One subtlety of de Finetti's views is that the requirement of coherence
only implies finite additivity and not countable additivity
(see, e.g., \citealt[Sect.~18.3]{deFinetti:2017}, and \citealt[Sect.~3.5.2]{Bernardo/Smith:2000}).
In the finite case, however,
the difference between finite and countable additivity disappears.
(Starting from a finite case is standard in probability theory;
see, e.g., \citealt[Chap.~I]{Kolmogorov:1933-Latin}, and \citealt[Chap.~1]{Shiryaev:2016}.)

Coherence is a property of consistency of the Bayesian's beliefs,
so we could call it internal coherence.
The English word ``coherence'' may cover not only internal coherence
but also agreement with reality (a kind of external coherence).
For example,
one of the earliest abstract uses of ``coherence'' given in the Oxford English Dictionary
(as of March 2024) is from Abraham Fraunce's ``The lawiers logike'' (1588):
``Where there is a greater coh{\ae}rence and affinitie betweene the argument and the thing argued''.
In this paper, however, we will use ``coherence'' in its meaning of internal coherence,
which is standard in Bayesian statistics,
except for a short discussion in Sect.~\ref{subsec:weakness}.

In this paper I will ignore any distinctions that are sometimes made
between ``forecast'' and ``prediction''
(such as predictions being more categorical than forecasts)
and will regard these words as synonymous.
I will never use more exotic words such as ``prevision''
\citep[Sect.~3.1.2]{deFinetti:2017}.

Following \citet{Shafer/Vovk:2001,Shafer/Vovk:2019} I use ``game-theoretic''
to refer to being based on betting,
and the kind of game theory involved here is the theory of perfect-information games
rather than the probabilistic games studied in, e.g., economics
(see, e.g., \citealt[Sect.~4.5]{Shafer/Vovk:2019}).

\subsection{Dramatis personae}

These are the players in our prediction protocols
(most of the protocols involve subsets of players).
\begin{itemize}
\item
  Reality (female):
  player who chooses sequential observations $y_1,y_2,\dots$,
  which are elements of the observation space $\mathbf{Y}$
  (assumed finite).
\item
  Forecaster (male):
  player who issues probabilistic forecasts for the future observations.
\item
  Sceptic (male):
  player who gambles against Forecaster's predictions.
  Informally, he is trying to discredit Forecaster.
\item
  Decision Maker (female):
  player who makes decisions in light of Forecaster's predictions.
\end{itemize}
The players' sexes are defined in \citet{Shafer/Vovk:2019}.
The noun ``Bayesian'' will often be used as nearly synonymous with ``Forecaster'',
and so the Bayesian is male.

\section{Basic prediction picture}
\label{sec:basic}

We are interested in the following sequential \emph{Bayesian prediction protocol}.

\begin{protocol}\label{prot:joint}\ \\
  \indentI FOR $n=1,\dots,N$:\\
    \indentII Forecaster announces $P_n\in\mathfrak{P}(\mathbf{Y}^{N-n+1})$\\
    \indentII Reality announces the actual observation $y_n\in\mathbf{Y}$.
\end{protocol}

\noindent
In this paper we only consider the case of a finite \emph{time horizon} $N>1$.
At each step $n$,
$P_n$ is a prediction for the whole future $y_n y_{n+1}\dots y_N$
(sequence of length $N-n+1$).
Earlier we referred to Forecaster as ``Bayesian'',
in order to emphasize that his predictions are complete probability measures
over the future observations
(while in earlier work we often considered less complete predictions:
see, e.g., \citealt[Preface, point 2]{Shafer/Vovk:2019}).

Protocol~\ref{prot:joint} does not define a game,
since we have not specified the players' goals,
but we will often talk about the \emph{plays}
$P_1 y_1 \dots P_N y_N$ proceeding according to the protocol's rules.

Let us assume, for simplicity,
that the set $\mathbf{Y}$ is finite
(and equipped with the discrete $\sigma$-algebra);
this will allow us to concentrate of conceptual issues
avoiding technical difficulties and ambiguities
(such as countable vs finite additivity).
To exclude trivialities, we also assume $\lvert\mathbf{Y}\rvert>1$.

In addition, we impose the requirement that $P_n(E)>0$ unless $E=\emptyset$.
This is a version of Lindley's ``Cromwell's rule'' \citep[Sect.~6.7]{Lindley:1985}.

\subsection{Bayesian conditioning}

Protocol~\ref{prot:joint} goes beyond \emph{Bayesian conditioning},
where we insist that, for each $n\ge2$,
\begin{equation}\label{eq:Bayesian-conditioning}
  P_{n}(x)
  =
  P_{n-1}(x\mid y_{n-1})
  :=
  \frac{P_{n-1}(y_{n-1}x)}{P_{n-1}(y_{n-1})},
  \quad
  x\in\mathbf{Y}^{N-n+1}.
\end{equation}
Bayesian conditioning (as rule for updating beliefs)
was criticised by \citet{Hacking:1967},
since the rule ignores the cost of thinking.
\citet{Lewis:1999} points out that
``we should sometimes respond to conceptual discoveries by revising our beliefs''.
However, the most straightforward reason for violating \eqref{eq:Bayesian-conditioning}
is that at step $n$ Forecaster can also learn other information apart from $y_n$
(i.e., learn information outside the protocol);
see \citet{Shafer:1985}.

Let us give an example where Bayesian conditioning,
based on the principle of total evidence,
is utterly unrealistic as an updating rule:
we can't hope to have a comprehensive protocol
including all the information a real-life Bayesian has access to.
Consider the standard case \citep{Dawid:1982Bayesian,Dawid:2006} of a weather forecaster
who issues a probability for the rain on sequentially numbered days.
The observations are the actual outcomes,
say 0 or 1 (encoding a dry or rainy day).
In the morning of day 1 the forecaster announces a joint probability
for the future observations (for days $1,2,\dots$) as his forecast,
and in the morning of day 2 he announces a new forecast, for days $2,3,\dots$.
We can't assume that the 0/1 observation on day 1 is all the extra information
that he has in the morning of day 2:
a serious weather forecaster, such as the UK Met Office,
will have plenty of other information arriving from weather stations around the globe
(and even from outer space).
This is a common situation;
to quote \citet[Sect.~4]{Goldstein:1983},
``in most cases of interest (e.g., the doctor's examination of the patient)
it is unreasonable to suppose that, even in principle,
there is a partition of possibilities
over which probabilities and conditional probabilities could, in theory, be defined.''
We will sometimes use the notation $\FFF_{n-1}$ (formally this is a $\sigma$-algebra)
for the information available when making the prediction $P_n$ at time $n$,
albeit in many cases this will be an unmanageable notion
that is even difficult to imagine
(while the moves in our protocols will be observable).

\begin{remark}\label{rem:K}
  Obviously,
  we can't include the data arriving from weather stations around the globe
  in a realistic prediction protocol,
  but we can go further and argue that even our picture is unrealistic
  for a large time horizon $N$:
  e.g., the first probability measure $P_1\in\mathfrak{P}(\mathbf{Y}^N)$ specifies
  $\lvert\mathbf{Y}\rvert^N-1$
  independent parameters,
  and this number grows exponentially in $N$ even for $\lvert\mathbf{Y}\rvert=2$.
  In Sect.~\ref{sec:K} we consider a more realistic setting
  of forecasting $K$ steps ahead
  (such as a week ahead for $K=7$).
\end{remark}

\begin{remark}
  Another reason why we might want to consider a Bayesian
  violating Bayesian conditioning when updating his beliefs
  is that his computational resources might be limited:
  he might keep processing information already available at the previous steps
  obtaining new values for probabilities of the same events.
  This is a special case of Hacking's (\citeyear{Hacking:1967}) observation mentioned earlier.
\end{remark}

\subsection{Weakness of coherence in the diachronic picture}
\label{subsec:weakness}

The following proposition
(proved in Sect.~\ref{subsec:coherence-useless} of Appendix~\ref{app:proofs})
shows that there is no reason to expect
any connection between the forecasts $P_n$
in Protocol~\ref{prot:joint}
if we only assume a natural diachronic modification of coherence.
Essentially this modification was used by \citet[p.~406]{Lewis:1999};
lack of coherence becomes, in Lewis's words,
``a risk of loss uncompensated by any chance of gain''.

\begin{proposition}\label{prop:coherence-useless}
  For any sequence of outcomes $y_1,\dots,y_N\in\mathbf{Y}$ and
  any sequence of probability measures $P_n\in\mathfrak{P}(\mathbf{Y}^{N-n+1})$,
  $n=1,\dots,N$,
  there is a positive finite probability space $(\Omega,P)$
  with filtration $\FFF_0,\dots,\FFF_N$
  and an adapted sequence of $\mathbf{Y}$-valued random elements $Y_1,\dots,Y_n$
  such that the event
  \[
    \forall n\in\{1,\dots,N\}:
    P_n=P(\cdot\mid\FFF_{n-1}) \;\&\; Y_n=y_n
  \]
  has a positive probability.
\end{proposition}

Remember that the Bayesian is incoherent
if we can set a system of bets under which he always loses
(a Dutch book).
This notion is not applicable in the diachronic setting
since we learn the full sequence $P_1,\dots,P_N$
only at the very end of the forecasting session, when it is too late to bet.
But we can apply Lewis's modification.
Let us say that the Bayesian (Forecaster in our protocol)
is \emph{dynamically incoherent} in a particular play
if there is a way of betting against him that never leads to Sceptic's loss
but, in this particular play, leads to Sceptic's gain.
In other words, it's a gain that is not compensated by a potential loss;
we will call it a \emph{gratis gain}.

Proposition~\ref{prop:coherence-useless} says that Forecaster
is never dynamically incoherent
provided the bets are fair under every possible $P$.
This is true under any strategy for probability updating
(or in the absence of such a strategy).
There cannot be any diachronic inconsistency between $P_n$ for different $n$
leading to a gratis gain for Sceptic.
In the following section we will see that such inconsistency
can lead to an \textbf{almost} gratis gain for Sceptic.
See also Remark~\ref{rem:hidden} below.

\section{Testing probability forecasts}
\label{sec:test}

Long-term prediction is much more complicated
than one-step-ahead prediction,
and to have a clear understanding of the process we will use two pictures,
which we call game-theoretic and measure-theoretic (following \citealt{Shafer/Vovk:2019}).
The game-theoretic picture is based on betting,
as in \citet{deFinetti:1937}.
In this section we discuss the game-theoretic picture,
and in the next one (Sect.~\ref{sec:embedding})
we move on to the measure-theoretic picture.

In the game-theoretic picture we add a third player, Sceptic, to the basic forecasting protocol.
``Sceptic'' is just our name for the bettor,
and betting proceeds according to the following protocol
(the intuition behind this protocol will be explained shortly).

\begin{protocol}\label{prot:joint-test-final}\ \\
  \indentI $\K_0 := 1$\\
  \indentI FOR $n=1,\dots,N$:\\
    \indentII Forecaster announces $P_n\in\mathfrak{P}(\mathbf{Y}^{N-n+1})$\\
    \indentII IF $n>1$:\\
    \indentIII $\K_{n-1} := \K_{n-2}
        + \sum_{x\in\mathbf{Y}^{N-n+1}} f_{n-1}(y_{n-1}x) P_{n}(x)$\\
        \indentV ${}- \sum_{x\in\mathbf{Y}^{N-n+2}} f_{n-1}(x) P_{n-1}(x)$%
        \hfill\refstepcounter{equation}(\theequation)\label{eq:joint-K-P-2}\\
    \indentII Sceptic announces $f_n\in\R^{\mathbf{Y}^{N-n+1}}$\\
    \indentII Reality announces $y_n\in\mathbf{Y}$\\
    \indentII IF $n=N$:\\
      \indentIII $\K_n := \K_{n-1} + f_n(y_n) - \sum_{y\in\mathbf{Y}} f_n(y) P_n(y)$.
          \hfill\refstepcounter{equation}(\theequation)\label{eq:joint-K-y-2}
\end{protocol}

\noindent
Sceptic's capital is not allowed to become negative
(as soon as it does, the play is stopped and Sceptic loses).
We regard this protocol (and similar protocols below)
as a way of testing Forecaster's predictions:
a large $\K_N$ means lack of agreement of his predictions with reality.
When for a particular play $\K_N$ is large,
we can regard it as an almost gratis gain.
(This assumes that Sceptic's capital is measured in small monetary units,
but we can always scale it up if the monetary units are not small.)

The interpretation of Protocol~\ref{prot:joint-test-final}
is that at each step $n$ Forecaster announces,
in the spirit of \citet[Chap.~3]{deFinetti:2017},
the price $P_n(x)$ for the uncertain quantity
$1_{\{(y_n,\dots,y_N)=x\}}$
for each $x\in\mathbf{Y}^{N-n+1}$.
We imagine a ticket that pays $1_{\{(y_n,\dots,y_N)=x\}}$
at the end of the play,
and so $P_n(x)$ is Forecaster's price
(at which he prepared to sell and to buy) for this ticket,
which we will call \emph{ticket $x$}.
After Forecaster's move $P_n$ Sceptic buys,
for each $x\in\mathbf{Y}^{N-n+1}$,
some tickets from Forecaster,
and $f_n(x)$ stands for the number of tickets $x$ that he chooses to buy
(the number can be positive, negative, or zero).
Sceptic pays $\sum_{x\in\mathbf{Y}^{N-n+1}} f_{n}(x) P_{n}(x)$
for the transaction.
If $n$ is not the last step,
at the next step Forecaster announces a new $P_n$,
and Sceptic sells all his tickets at the new prices.
Then Sceptic buys a new set of tickets at the new prices, etc.
(An important special case is where the new set of tickets
coincides with the old set,
so effectively no trade happens and Sceptic just keeps the old set.)
The change in his capital at step $n<N$ is summarized in \eqref{eq:joint-K-P-2}:
\begin{itemize}
\item
  $\sum_{x\in\mathbf{Y}^{N-n+1}} f_{n-1}(y_{n-1}x) P_{n}(x)$
  is the amount he gains at this step
  by selling, at the current prices $P_n$, the tickets that he bought at the previous step;
  notice that only tickets $y_{n-1}x$ will have non-zero prices;
\item
  $\sum_{x\in\mathbf{Y}^{N-n+2}} f_{n-1}(x) P_{n-1}(x)$ is the amount
  paid for those tickets at the previous step.
\end{itemize}
At the last step he just cashes in the winning ticket $y_N$:
see \eqref{eq:joint-K-y-2}.

\subsection{Comparison with Bayesian conditioning}

How does \eqref{eq:joint-K-P-2} compare with Bayesian conditioning,
where no new information outside the protocol arrives
and we just define $P_n$ by \eqref{eq:Bayesian-conditioning}?
In this case we can simplify the protocol
by replacing Sceptic's moves $f_n$ with $f'_n:\mathbf{Y}\to\R$ defined by
\begin{equation}\label{eq:f'_n}
  f'_n(y)
  :=
  \sum_{x\in\mathbf{Y}^{N-n}}
  f_n(y x)
  P_n(x\mid y),
\end{equation}
and then \eqref{eq:joint-K-P-2} becomes
\[
  \K_{n-1}
  :=
  \K_{n-2}
  +
  f'_{n-1}(y_{n-1})
  -
  \sum_{y\in\mathbf{Y}} f'_{n-1}(y) P_{n-1}(y).
\]
Moving this command to the previous step,
we can rewrite Protocol~\ref{prot:joint-test-final} as
\begin{protocol}\label{prot:joint-test-Bayesian}\ \\
  \indentI $\K_0 := 1$\\
  \indentI FOR $n=1,\dots,N$:\\
    \indentII IF $n=1$:\\
      \indentIII Forecaster announces $P_1\in\mathfrak{P}(\mathbf{Y}^N)$\\
    \indentII ELSE:\\
      \indentIII Forecaster updates $P_{n-1}\in\mathfrak{P}(\mathbf{Y}^{N-n+2})$
        to $P_n\in\mathfrak{P}(\mathbf{Y}^{N-n+1})$\\
        \indentV by Bayesian conditioning~\eqref{eq:Bayesian-conditioning}\\
    \indentII Sceptic announces $f'_n\in\R^{\mathbf{Y}}$\\
    \indentII Reality announces $y_n\in\mathbf{Y}$\\
    \indentII $\K_n := \K_{n-1} + f'_n(y_n) - \sum_{y\in\mathbf{Y}} f'_n(y) P_n(y)$.
\end{protocol}
This is our standard one-step-ahead prediction protocol
(cf., e.g., \citealt[Protocol~1.1]{Shafer/Vovk:2019})
except that Forecaster announces his forecasting strategy in advance.
We can see that forecasting multiple steps ahead
does not require any new methods under Bayesian conditioning:
testing can proceed one step ahead.

Without any restrictions on Forecaster,
we obtain, instead of Protocol~\ref{prot:joint-test-Bayesian},
the following protocol equivalent to Protocol~\ref{prot:joint-test-final},
in which we still use the notation \eqref{eq:f'_n}.

\begin{protocol}\label{prot:conditioning-prot}\ \\
  \indentI $\K_0 := 1$\\
  \indentI FOR $n=1,\dots,N$:\\
    \indentII Forecaster announces $P_n\in\mathfrak{P}(\mathbf{Y}^{N-n+1})$\\
    \indentII IF $n>1$:\\
      \indentIII  $\K_{n-1} := \K_{n-1}
        + \sum_{x\in\mathbf{Y}^{N-n+1}} f_{n-1}(y_{n-1}x) (P_{n}(x)-P_{n-1}(x\mid y_{n-1}))$\\
    \indentII Sceptic announces $f_n\in\R^{\mathbf{Y}^{N-n+1}}$\\
    \indentII Reality announces $y_n\in\mathbf{Y}$\\
    \indentII $\K_n := \K_{n-1} + f'_n(y_n) - \sum_{y\in\mathbf{Y}} f'_n(y) P_n(y)$.
\end{protocol}

Protocol~\ref{prot:conditioning-prot} adds to Protocol~\ref{prot:joint-test-Bayesian}
the possibility to update $P_n$ is a way different from Bayesian conditioning
and includes a term that describes betting on the difference
between the actual forecast $P_n(x)$
and the Bayesian conditional probabilities $P_{n-1}(x\mid y_{n-1})$
computed from the previous forecast.
The equivalence of Protocols~\ref{prot:joint-test-final} and~\ref{prot:conditioning-prot}
follows from the equality, for $n<N$,
of the addend $f'_n(y_n)$ in the expression for $\K_n$
in Protocol~\ref{prot:conditioning-prot}
and the subtrahend
$\sum_{x\in\mathbf{Y}^{N-n+1}} f_{n-1}(y_{n-1}x) P_{n-1}(x\mid y_{n-1})$
in the expression for $\K_{n-1}$ at the next step.

\subsection{Merging Sceptic's opponents}

If we are only interested in strategies for Sceptic
(not in strategies for other players,
as in \citealt[Preface, ideas 3 and 6]{Shafer/Vovk:2019})
we can simplify Protocol~\ref{prot:joint-test-final} further
by merging Forecaster and Reality.
We will refer to the combined player as Forecaster
(rather than World, as in \citealt{Shafer/Vovk:2001,Shafer/Vovk:2019});
the reason for this will become clear
in Sect.~\ref{subsec:radical-additive} in Appendix~\ref{app:radical}.

\begin{protocol}\label{prot:merged}\ \\
  \indentI $\K_0 := 1$\\
  \indentI FOR $n=1,\dots,N,N+1$:\\
    \indentII Forecaster announces $Q_n\in\mathfrak{P}_n(\mathbf{Y}^{N})$\\
    \indentII IF $n>1$:\\
    \indentIII $\K_{n-1} := \K_{n-2}
        + \sum_{x\in\mathbf{Y}^{N}} F_{n-1}(x) (Q_{n}(x) - Q_{n-1}(x))$
        \text{\quad}\hfill\refstepcounter{equation}(\theequation)\label{eq:merged}\\
    \indentII Sceptic announces $F_n\in\R^{\mathbf{Y}^{N}}$.
\end{protocol}

\noindent
Protocol~\ref{prot:merged} uses the notation $\mathfrak{P}_n(\mathbf{Y}^{N})$
for the set of all probability measures $Q$ on $\mathbf{Y}^{N}$
satisfying $Q(x)=1$ for some $x\in\mathbf{Y}^{n-1}$.

To embed Protocol~\ref{prot:joint-test-final} into Protocol~\ref{prot:merged},
we should take as $Q_n$ the extension of $P_n$ to $\mathbf{Y}^N$,
namely
\begin{equation}\label{eq:Q}
  Q_n(x)
  :=
  \begin{cases}
    P_n(x\setminus(y_1\dots y_{n-1})) & \text{if $(y_1\dots y_{n-1})\subseteq x$}\\
    0 & \text{otherwise},
  \end{cases}
\end{equation}
and we should take as $F_n$ the extension of $f_n$ to $\mathbf{Y}^N$,
namely
\[
  F_n(x)
  :=
  \begin{cases}
    f_n(x\setminus(y_1\dots y_{n-1})) & \text{if $(y_1\dots y_{n-1})\subseteq x$}\\
    u & \text{otherwise},
  \end{cases}
\]
where, e.g., $u:=0$
(but in fact the value of $u$ does not matter as it is always multiplied by 0
in the embedded protocol,
and we can use different $u$s for different $n$ and $x$).

Protocol~\ref{prot:merged} lasts for $N+1$ rather than $N$ steps
in order for $\K_N$ to be defined by \eqref{eq:merged}.
Sceptic's last move $F_{N+1}$ is never used.

\section{Measure-theoretic picture}
\label{sec:embedding}

This section may be skimmed or skipped completely;
the remaining sections do not depend on it.

The measure-theoretic picture is stochastic
and assumes an overall probability measure
used by Forecaster and Reality for generating their moves.
In other words, it is just the standard measure-theoretic framework
\citep{Kolmogorov:1933-Latin,Doob:1953}
for probability.
In this section we will define testing in the measure-theoretic picture
and will see that it is equivalent to (albeit more complicated and less natural than)
testing in the game-theoretic picture,
i.e., the testing procedure described in the previous section.

Our check of equivalence will have two sides:
the validity of the game-theoretic picture in the measure-theoretic framework,
and the validity of a natural measure-theoretic picture in the game-theoretic framework.
In our proofs in Appendix~\ref{app:proofs} (Sect.~\ref{subsec:validity})
we will use a standard theorem of duality;
in general, it can be said that the measure-theoretic and game-theoretic pictures
are dual to each other in a certain sense.

What exactly do I mean by equivalence?
The idea is to show that we have identical ways of gambling
in both pictures.
On the measure-theoretic side,
we have the standard notion of measure-theoretic martingale
(defined later in the section),
and we define a \emph{test martingale} as nonnegative martingale
with initial value 1.
On the game-theoretic side,
a \emph{game-theoretic test martingale} is Sceptic's capital $\K_n$
(for all possible $n$),
for a fixed strategy for Sceptic,
as function of Forecaster's and Reality's moves
provided this function is nonnegative.
Roughly, the equivalence means the equivalence of the two notions of test martingale,
but the exact statement will become clear in the process of its demonstration
(when we reach Proposition~\ref{prop:validity}).

\subsection{Validity of the game-theoretic picture}
\label{subsec:GTP-valid}

Let $(\Omega,P)$ be a finite probability space
equipped with a filtration $\FFF_n$, $n=0,1,\dots$.
Intuitively, we regard $\FFF_{n-1}$ as the information
available to Forecaster and Sceptic at the beginning of step $n$
in Protocol~\ref{prot:joint-test-final}.
For concreteness, let us assume
that all new information (including $y_n$, which is part of the new information)
arrives at the end of step $n$
and none arrives between the steps;
therefore, $\FFF_n$, $n=1,2,\dots$, is the information available at the end of step $n$.

In the measure-theoretic framework for Protocol~\ref{prot:joint-test-final},
we assume that $y_1,\dots,y_N$ are realizations
of an adapted $\mathbf{Y}$-valued process $Y_1,\dots,Y_N$
(meaning, as usual, that $Y_n$ is $\FFF_n$-measurable, $n=1,\dots,N$),
that each $P_n$ is computed from $P$ as the conditional probability measure for $Y_n$
given $\FFF_{n-1}$,
and that Sceptic follows a predictable strategy
(where ``predictable'' has the technical meaning that $f_n(x)$ is $\FFF_{n-1}$-measurable for each $x$).
Sceptic's capital $\K_n$ is then an adapted process,
and we have
\begin{equation}\label{eq:P_n}
  P_n(x)
  =
  P
  \left(
    \{Y_n\dots Y_N=x\}
    \mid
    \FFF_{n-1}
  \right)
  \quad \text{a.s.},
  \quad
  x\in\mathbf{Y}^{N-n+1}.
\end{equation}
Now the increment \eqref{eq:joint-K-P-2} in Sceptic's capital is
\begin{align*}
  \K_{n-1}-\K_{n-2}
  &=
  \sum_{x\in\mathbf{Y}^{N-n+1}}
  f_{n-1}(Y_{n-1}x)
  P(\{(Y_n,\dots,Y_N)=x\}\mid\FFF_{n-1})\\
  &\quad-
  \sum_{x\in\mathbf{Y}^{N-n+2}}
  f_{n-1}(x)
  P(\{(Y_{n-1},\dots,Y_N)=x\}\mid\FFF_{n-2})\\
  &=
  \E_P
  \left(
    f_{n-1}(Y_{n-1},\dots,Y_N)\mid\FFF_{n-1}
  \right)
  -
  \E_P
  \left(
    f_{n-1}(Y_{n-1},\dots,Y_N)\mid\FFF_{n-2}
  \right),
\end{align*}
and so we have
\begin{equation}\label{eq:E-Delta-1}
  \E_P(\K_{n-1}-\K_{n-2}\mid\FFF_{n-2})
  =
  0
  \quad
  \text{a.s.}.
\end{equation}
The exceptional (for $n=N$) increment \eqref{eq:joint-K-y-2} is
\begin{equation*}
  \K_n-\K_{n-1}
  =
  f_n(Y_n)
  -
  \E_P(f_n(Y_n)\mid\FFF_{n-1}),
\end{equation*}
which gives the analogue
\begin{equation*}
  \E_P(\K_n-\K_{n-1}\mid\FFF_{n-1})
  =
  0
  \quad
  \text{a.s.}
\end{equation*}
of \eqref{eq:E-Delta-1}.
Therefore, $\K_0,\K_1,\dots,\K_N$
is a measure-theoretic martingale w.r.\ to the filtration $(\FFF_n)$:
$\E_P(\K_n\mid\FFF_{n-1})=\K_{n-1}$, $n=1,\dots,N$.

\subsection{Validity of the measure-theoretic picture}

A \emph{non-terminal situation} in Protocol~\ref{prot:joint-test-final}
(and also in Protocol~\ref{prot:joint})
is a tuple $(P_1,y_1,\dots,P_n)$ for some $n\in\{1,\dots,N\}$,
where $y_i\in\mathbf{Y}$ and $P_i\in\mathfrak{P}(\mathbf{Y}^{N-i+1})$ for all $i$.
Informally, this is a situation in which Sceptic makes a move.
A \emph{terminal situation}
is a tuple $(P_1,y_1,\dots,P_N,y_N)$,
where again $y_i\in\mathbf{Y}$ and $P_i\in\mathfrak{P}(\mathbf{Y}^{N-i+1})$ for all $i$.
Non-terminal situations and terminal situations
are referred to collectively as \emph{situations}.
A strategy for Sceptic can be defined as a function
mapping the non-terminal situations to an allowed move,
namely mapping a situation $(P_1,y_1,\dots,P_n)$
to $f\in\R^{\mathbf{Y}^{N-n+1}}$ in the case of Protocol~\ref{prot:joint-test-final}.
For a fixed strategy for Sceptic
his capital becomes a real-valued function of a situation;
let us refer to such functions as \emph{game-theoretic test martingales}
provided they are nonnegative.

A \emph{game-theoretic process} is a Borel measurable real-valued function of a situation.
A nonnegative game-theoretic process $S$ is a \emph{visible measure-theoretic test martingale}
if, for any finite probability space $(\Omega,P)$ equipped with a filtration $(\FFF_n)_{n=0}^N$
and any adapted sequence of random variables $Y_1,\dots,Y_N$,
\begin{equation}\label{eq:S}
  \begin{aligned}
    S_{n}
    &:=
    S(P_1,Y_1,\dots,P_{n+1}),
    \quad
    n=0,\dots,N-1,\\
    S_N
    &:=
    S(P_1,Y_1,\dots,P_N,Y_N)
  \end{aligned}
\end{equation}
is a test martingale in the usual sense of $S_0=1$ and
\begin{equation}\label{eq:martingale}
  \E_P(S_n\mid\FFF_{n-1})
  =
  S_{n-1},
  \quad
  n=1,\dots,N,
\end{equation}
where the $P_i$ in \eqref{eq:S} are defined by \eqref{eq:P_n},
which becomes
\[
  P_i(x)
  :=
  P
  \left(
    \{Y_i\dots Y_N=x\}
    \mid
    \FFF_{i-1}
  \right),
  \quad
  x\in\mathbf{Y}^{N-i+1},
\]
in our current notation.
The adjective ``visible'' refers to the martingale $(S_n)$
depending only on the players' moves in Protocol~\ref{prot:joint}
(and not depending on the hidden aspects of the realized sample point $\omega\in\Omega$).

The following statement of agreement between the game-theoretic and measure-theoretic pictures
will be proved in Sect.~\ref{subsec:validity}.

\begin{proposition}\label{prop:validity}
  A game-theoretic process is a game-theoretic test martingale
  if and only if it is a visible measure-theoretic test martingale.
\end{proposition}

Proposition~\ref{prop:validity}, however, has a weakness.
Let us say that a game-theoretic process is a \emph{game-theoretic test supermartingale}
if it can be obtained as Sceptic's capital
while he is allowed to discard part of his capital at each step
(but is still not allowed to go into debt).
For example, in the case of Protocol~\ref{prot:joint-test-final}
this corresponds to replacing \eqref{eq:joint-K-P-2} and \eqref{eq:joint-K-y-2}
by Sceptic's moves allowing him to choose $\K_{n-1}$ and $\K_n$,
respectively,
as any nonnegative number not exceeding the corresponding right-hand side.
And a game-theoretic process is a \emph{visible measure-theoretic test supermartingale}
if it is defined in the same way as a visible measure-theoretic test martingale
except that the ``$=$'' in \eqref{eq:martingale} is replaced by ``$\le$''.
The notion of a game-theoretic test supermartingale is obviously redundant,
in the sense of every game-theoretic test supermartingale
being dominated by a game-theoretic test martingale.
But the requirement \eqref{eq:martingale} holding for any finite probability space
might appear restrictive,
and so it is less obvious that measure-theoretic test supermartingales
are redundant in this sense.
Therefore, in Sect.~\ref{subsec:validity} we will start from proving
the following modification of Proposition~\ref{prop:validity}.

\begin{theorem}\label{thm:validity}
  A game-theoretic process is a game-theoretic test supermartingale
  if and only if it is a visible measure-theoretic test supermartingale.
\end{theorem}

This theorem implies that every visible measure-theoretic test supermartingale
is dominated by a visible measure-theoretic test martingale.
We will also check this directly in Sect.~\ref{subsec:validity}.

\begin{remark}\label{rem:hidden}
  This section, and Sect.~\ref{subsec:weakness} above,
  illustrate the ``hidden variable'' account of belief change
  (\citealp[Chap.~4, note 14]{Adams:1975},
  \citealp[Theorem 2.1]{Diaconis/Zabell:1982}, \citealp[Sect.~1]{Skyrms:1992}),
  according to which coherent belief update is Bayesian conditioning
  in a bigger belief space.
\end{remark}

\section{Predicting $K$ steps ahead}
\label{sec:K}

For a large time horizon $N$,
the protocols considered in the previous sections are unrealistic
in that Forecaster is asked to produce probability measures
on huge sets such as $\mathbf{Y}^N$.
Starting from this section,
we will assume that all predictions made by Forecaster
are only for the next $K<N$ observations, with $K\ge1$,
and we will sometimes refer to $K$ as the \emph{prediction horizon}.
We are typically interested in the case $K\ll N$.

We can still use the Bayesian prediction protocol (Protocol~\ref{prot:joint}),
but now Sceptic is not allowed to bet more than $K$ steps ahead.
In terms of Protocol~\ref{prot:joint-test-final},
the function $f_n\in\R^{\mathbf{Y}^{N-n+1}}$
depends on its argument $(y_n,\dots,y_N)$
only via its first $K$ elements $y_n,\dots,y_{n+K-1}$
(let us assume for the moment that $n+K-1\le N$).
Writing $f_n(y_n,\dots,y_{n+K-1})$ instead of $f_n(y_n,\dots,y_N)$,
we obtain the following modification of Protocol~\ref{prot:joint-test-final}.

\begin{protocol}\label{prot:joint-test-final-K}\ \\
  \indentI $\K_0 := 1$\\
  \indentI FOR $n=1,\dots,N$:\\
    \indentII Forecaster announces $P_n\in\mathfrak{P}(\mathbf{Y}^{N-n+1})$\\
    \indentII IF $n>1$:\\
      \indentIII $\K_{n-1} := \K_{n-2}
          + \sum_{x\in\mathbf{Y}^{(K-1)\wedge(N-n+1)}} f_{n-1}(y_{n-1}x) P_{n}(x)$\\
          \indentV ${}- \sum_{x\in\mathbf{Y}^{K\wedge(N-n+2)}} f_{n-1}(x) P_{n-1}(x)$\\
    \indentII Sceptic announces $f_n\in\R^{\mathbf{Y}^{K\wedge(N-n+1)}}$\\
    \indentII Reality announces $y_n\in\mathbf{Y}$\\
    \indentII IF $n=N$:\\
      \indentIII $\K_n := \K_{n-1} + f_n(y_n) - \sum_{y\in\mathbf{Y}} f_n(y) P_n(y)$.
\end{protocol}

\noindent
Of course, we obtain an equivalent protocol
if we replace $P_n\in\mathfrak{P}(\mathbf{Y}^{N-n+1})$
by $P_n\in\mathfrak{P}(\mathbf{Y}^{K\wedge(N-n+1)})$
in the third line,
and this replacement would eliminate an irrelevant part of $P_n$.
Alternatively, we obtain an equivalent protocol if we require
$P_n\in\mathfrak{P}(\mathbf{Y}^{K})$.

\begin{remark}
  In the example of weather forecasting one week ahead
  (cf.\ Remark~\ref{rem:K}),
  the predictions $P_n\in\mathfrak{P}(\mathbf{Y}^{K})$
  are quite different from the predictions produced by a typical weather app.
  Weather apps produce marginal probabilities of rain
  whereas the probabilities in $P_n\in\mathfrak{P}(\mathbf{Y}^{K})$
  are joint.
  Testing marginal probabilities would be easier
  than the kind of testing exemplified
  by Protocol~\ref{prot:joint-test-final-K}.
  See \citet{Vovk:Logic-v1} for details of testing marginal probabilities.
\end{remark}

\section{Bayesian decision making}
\label{sec:DM}

Why do we need long-term forecasts?
One reason is that they facilitate nearly optimal decisions.

\subsection{An optimality result for the Bayes decision strategy}

Consider the following decision-making protocol.

\begin{protocol}\label{prot:DM-1}\ \\
  \indentI FOR $n=1,\dots,N$:\\
    \indentII Reality announces $\lambda_n:\mathbf{D}\times\mathbf{Y}^{N-n+1}\to[0,1]$\\
    \indentII Decision Maker announces $d_n\in\mathbf{D}$\\
    \indentII Reality announces the actual observation $y_n\in\mathbf{Y}$.
\end{protocol}

\noindent
At each step $n$ Decision Maker is asked to choose a decision $d_n$
from a finite set $\mathbf{D}$ of permitted decisions.
Before that,
Reality announces a loss function $\lambda_n$ determining Decision Maker's loss
\[
  \lambda_n(d_n,y_n\dots y_N)\in[0,1]
\]
at this step.
In applications the loss functions are usually given in advance,
but we include them in the protocol in order to weaken the conditions of our mathematical result
(Theorem~\ref{thm:optimal} below).
The loss functions are assumed bounded and scaled to the interval $[0,1]$.
The total loss can be computed only after the last step and equals
\begin{equation}\label{eq:Loss}
  \Loss_N
  :=
  \sum_{n=1}^N
  \lambda_n(d_n,y_n\dots y_N)
  \in
  [0,N].
\end{equation}
Of course, $\Loss_N$ is a function of Reality's and Decision Maker's moves,
but we will leave the arguments of $\Loss_N$ implicit.

A strategy for Decision Maker in Protocol~\ref{prot:DM-1}
is a function giving a decision $d_n$ at each step $n$
as function of Reality's previous moves $y_1,\dots,y_{n-1}$ and $\lambda_1,\dots,\lambda_n$.
It would be ideal to have a strategy $A$ for Decision Maker
that is provably either better than any other strategy $B$ or approximately equally good,
but this is clearly impossible;
we need a qualification of the type ``with high probability'',
and our decision making protocol is too poor to express it.

As a first step towards the goal of designing an optimal (in some sense)
strategy for Decision Maker,
we add a new player, Forecaster, to Protocol~\ref{prot:DM-1}.
The following protocol is a combination of Protocols~\ref{prot:DM-1} and~\ref{prot:joint}.

\begin{protocol}\label{prot:DM-2}\ \\
  \indentI FOR $n=1,\dots,N$:\\
    \indentII Reality announces $\lambda_n:\mathbf{D}\times\mathbf{Y}^{N-n+1}\to[0,1]$\\
    \indentII Forecaster announces $P_n\in\mathfrak{P}(\mathbf{Y}^{N-n+1})$\\
    \indentII Decision Maker announces $d_n\in\mathbf{D}$\\
    \indentII Reality announces the actual observation $y_n\in\mathbf{Y}$.
\end{protocol}

Protocol~\ref{prot:DM-2} allows us to design a plausible strategy
(\emph{Bayes strategy}, or \emph{Bayes optimal strategy})
for Decision Maker
(where $d_n$ is now allowed to depend, additionally, on Forecaster's previous moves
$P_1,\dots,P_n$):
\begin{equation}\label{eq:A}
  d_n
  \in
  \arg\min_{d\in\mathbf{D}}
  \sum_{x\in\mathbf{Y}^{N-n+1}}
  \lambda_n(d,x) P_n(x).
\end{equation}
However, we cannot prove anything about this strategy
as we do not know anything about connections between the forecasts $P_n$
and the actual observations $y_n$.
Therefore, we add Sceptic to our protocol,
as in Protocol~\ref{prot:joint-test-final}.

\begin{protocol}\label{prot:DM-3}\ \\
  \indentI $\K_0 := 1$\\
  \indentI FOR $n=1,\dots,N$:\\
    \indentII Reality announces $\lambda_n:\mathbf{D}\times\mathbf{Y}^{N-n+1}\to[0,1]$\\
    \indentII Forecaster announces $P_n\in\mathfrak{P}(\mathbf{Y}^{N-n+1})$\\
    \indentII IF $n>1$:\\
      \indentIII $\K_{n-1} := \K_{n-2}
          + \sum_{x\in\mathbf{Y}^{N-n+1}} f_{n-1}(y_{n-1}x) P_{n}(x)$\\
          \indentV ${}- \sum_{x\in\mathbf{Y}^{N-n+2}} f_{n-1}(x) P_{n-1}(x)$\\
    \indentII Decision Maker announces $d_n\in\mathbf{D}$\\
    \indentII Sceptic announces $f_n\in\R^{\mathbf{Y}^{N-n+1}}$\\
    \indentII Reality announces $y_n\in\mathbf{Y}$\\
    \indentII IF $n=N$:\\
      \indentIII $\K_n := \K_{n-1} + f_n(y_n) - \sum_{y\in\mathbf{Y}} f_n(y) P_n(y)$.
\end{protocol}

In order to prove a law of large numbers for decision making
showing that the Bayes strategy is indeed optimal in some sense,
we need the following combination of Protocols~\ref{prot:DM-3} and~\ref{prot:joint-test-final-K}
that only involves prediction $K$ steps ahead.
(We will see in Sect.~\ref{subsec:essential} that such a law of large numbers
inevitably fails for Protocol~\ref{prot:DM-3}.)

\begin{protocol}\label{prot:DM-4}\ \\
  \indentI $\K_0 := 1$\\
  \indentI FOR $n=1,\dots,N$:\\
    \indentII Reality announces $\lambda_n:\mathbf{D}\times\mathbf{Y}^{K}\to[0,1]$\\
    \indentII Forecaster announces $P_n\in\mathfrak{P}(\mathbf{Y}^{K})$\\
    \indentII IF $n>1$:\\
      \indentIII $\K_{n-1} := \K_{n-2}
          + \sum_{x\in\mathbf{Y}^{(K-1)\wedge(N-n+1)}} f_{n-1}(y_{n-1}x) P_{n}(x)$\\
          \indentV ${}- \sum_{x\in\mathbf{Y}^{K\wedge(N-n+2)}} f_{n-1}(x) P_{n-1}(x)$
          \hfill\refstepcounter{equation}(\theequation)\label{eq:joint-K-P-4}\\
    \indentII Decision Maker announces $d_n\in\mathbf{D}$\\
    \indentII Sceptic announces $f_n\in\R^{\mathbf{Y}^{K\wedge(N-n+1)}}$\\
    \indentII Reality announces $y_n\in\mathbf{Y}$\\
    \indentII IF $n=N$:\\
      \indentIII $\K_n := \K_{n-1} + f_n(y_n) - \sum_{y\in\mathbf{Y}} f_n(y) P_n(y)$.
\end{protocol}

We will continue to use the notation $\Loss_N$ introduced in \eqref{eq:Loss},
which is now modified to
\begin{equation}\label{eq:Loss-mod}
  \Loss_N
  :=
  \sum_{n=1}^{N-K+1}
  \lambda_n(d_n,y_n\dots y_{n+K-1}),
\end{equation}
but we will also be interested in Decision Maker's loss $\Loss_N(A)$
computed by replacing his actual decisions
by the decisions prescribed by a decision strategy~$A$:
\begin{equation*} 
  \Loss_N(A)
  :=
  \sum_{n=1}^{N-K+1}
  \lambda_n(d^A_n,y_n\dots y_{n+K-1}),
\end{equation*}
where
\[
  d^A_n
  :=
  A(\lambda_1,P_1,y_1,\lambda_2,P_2,\dots,y_{n-1},\lambda_n,P_n),
  \quad
  n=1,\dots,N-K+1;
\]
we are only interested in strategies that are functions of the previous moves
by Reality and Forecaster.
Let us adapt the Bayes strategy \eqref{eq:A} to Protocol~\ref{prot:DM-4}:
\begin{equation}\label{eq:AA}
  d_n
  :=
  d^A_n
  \in
  \arg\min_{d\in\mathbf{D}}
  \sum_{x\in\mathbf{Y}^{K}}
  \lambda_n(d,x) P_n(x),
\end{equation}
with $d^A_n$ chosen as the first element of the $\arg\min$ in a fixed linear order on $\mathbf{D}$
if there are ties among $d$.

If $E$ is a property of Reality's, Forecaster's, and Decision Maker's moves
in Protocol~\ref{prot:DM-4},
we define the \emph{upper game-theoretic probability} of $E$
as the infimum of $\alpha>0$ such that Sceptic has a strategy
that guarantees $\K_n\ge0$ for all $n$
and that ensures $\alpha\K_n\ge1$ whenever $E$ happens.
The following optimality result will be proved
in Appendix~\ref{app:proofs} (Sect.~\ref{subsec:optimal}).

\begin{theorem}\label{thm:optimal}
  Let $\epsilon\in(0,0.3)$.
  There is a strategy $A$ for Decision Maker in Protocol~\ref{prot:DM-4}
  that guarantees
  \begin{equation}\label{eq:optimal}
    \UP
    \left(
      \Loss_N(A) - \Loss_N
      \ge
      2\sqrt{K N \ln\frac{1}{\epsilon}}
    \right)
    \le
    \epsilon.
  \end{equation}
\end{theorem}

An alternative statement of Theorem~\ref{thm:optimal}
not using the notion of game-theoretic probability
is that there exists a joint strategy for Decision Maker and Sceptic that achieves either
\begin{equation}\label{eq:upper}
  \Loss_N(A) - \Loss_N
  <
  2\sqrt{K N \ln\frac{1}{\epsilon}}
\end{equation}
or $\K_N\ge1/\epsilon$.
For a small $\epsilon$ and large $N$ (as compared with $K\ln\frac{1}{\epsilon}$),
this joint strategy demonstrates that $A$ performs
better than or similarly to the actual moves $d_n$
unless Forecaster is discredited.
This is a version of the law of large numbers that works only when $K\ll N$.

\begin{remark}\label{rem:strong-law}
  Notice that the strong law of large numbers for a fixed $K$ (and with $N\to\infty$, as usual) is trivial:
  we can apply the standard one-step-ahead strong law of large numbers to each $K$th observation
  (starting from observation $1$,
  starting from observation $2$,\dots,
  and finally starting from observation $K$).
  Theorem~\ref{thm:optimal} is less trivial,
  but interestingly, it is based on the same idea.
  The argument used in the arXiv version~1 of this paper is different
  but leads to a weaker result (Theorem~7.5 in that version).
  See Remark~\ref{rem:inefficient} for further details.
\end{remark}

The strategy $A$ in the statement of Theorem~\ref{thm:optimal}
can be chosen as the Bayes optimal strategy~\eqref{eq:AA}.
Theorem~\ref{thm:optimal} shows that, for any other strategy $B$ for Sceptic,
we have
\begin{equation}\label{eq:optimal-1}
  \UP
  \left(
    \Loss_N(A) - \Loss_N(B)
    \ge
    2\sqrt{K N \ln\frac{1}{\epsilon}}
  \right)
  \le
  \epsilon;
\end{equation}
we, however, prefer the stronger statement \eqref{eq:optimal}
allowing Forecaster to choose his moves on the fly.
We can rewrite \eqref{eq:optimal-1} as
\begin{equation*}
  \UP
  \left(
    \frac1N \Loss_N(A)
    -
    \frac1N \Loss_N(B)
    \ge
    \delta
  \right)
  \le
  \exp
  \left(
    -\frac{\delta^2 N}{4K}
  \right)
\end{equation*}
for any $\delta\ge2.2\sqrt{K/N}$.
The restriction $\delta\ge2.2\sqrt{K/N}$
is coming from the condition $\epsilon<0.3$ in Theorem~\ref{thm:optimal};
without this restriction, we can still claim that
\begin{equation}\label{eq:no-restriction}
  \UP
  \left(
    \frac1N \Loss_N(A)
    -
    \frac1N \Loss_N(B)
    \ge
    \delta
  \right)
  \le
  5
  \frac{K}{\delta^2 N}
  \exp
  \left(
    -\frac{\delta^2 N}{4K}
  \right)
\end{equation}
(for a proof, see the end of Sect.~\ref{subsec:optimal}).

\begin{remark}
  In Theorem~\ref{thm:optimal} we compare Decision Maker's actual loss $\Loss_N$
  with the loss she would have suffered following the strategy $A$
  defined by \eqref{eq:AA}.
  Our interpretation of this theorem depends on the assumption
  that Reality's and Forecaster's moves are not affected by Decision Maker's moves.
\end{remark}

\subsection{Predicting $K<N$ steps ahead is essential for our statement of optimality}
\label{subsec:essential}

Theorem~\ref{thm:optimal} is about predicting $K$ steps ahead.
How important is this restriction?
Let us check that it may not be true that
\begin{equation}\label{eq:crude}
  \frac1N(\Loss_N(A) - \Loss_N)
  <
  \delta
\end{equation}
with high probability in Protocol~\ref{prot:DM-3} for $\delta\ll 1$
if we use the definition of the cumulative loss given in \eqref{eq:Loss}
(there is little difference between \eqref{eq:Loss} and \eqref{eq:Loss-mod}
for $K\ll N$, but for $K=N$ the latter leads
to vacuous statements for $\Loss_N(A) - \Loss_N$);
as before, $A$ stands for the Bayes optimal strategy.
The intuition behind this demonstration is that at each step
Decision Maker is asked to predict the last observation $y_N$,
and this creates heavy dependence between losses at different steps
that ruins the law of large numbers.

Set $\mathbf{D}:=\mathbf{Y}:=\{0,1\}$,
and suppose (in the spirit of measure-theoretic probability)
that all players know and comply with a probability measure
$P\in\mathfrak{P}(\{0,1\}^N)$
governing Reality.
The loss functions output by Reality are
\begin{equation}\label{eq:binary-loss}
  \lambda_n(d_n,y_n\dots y_N)
  :=
  \begin{cases}
    0 & \text{if $d_n=y_N$}\\
    1 & \text{otherwise},
  \end{cases}
\end{equation}
and the true probability measure $P$ is such that
$P(\{y_N=1\})=0.4$
(so that $y_N=0$ is slightly likelier than $y_N=1$).

The Bayes optimal strategy $A$ given by \eqref{eq:A} is $d_n^A:=0$.
Let us compare it with the complementary strategy $B:=1-A$
(or simply $B:=1$).
We have
\begin{equation}\label{eq:counter-example}
  \frac1N(\Loss_N(A)-\Loss_N(B))
  =
  \begin{cases}
    1 & \text{with probability $0.4$}\\
    -1 & \text{with probability $0.6$},
  \end{cases}
\end{equation}
and so the inequality~\eqref{eq:crude} is grossly violated
with a significant probability.

Applying the idea leading to \eqref{eq:counter-example}
on a smaller scale
(to each $K$th step instead of the last step),
we obtain the following lower bound for Protocol~\ref{prot:DM-4}.

\begin{proposition}\label{prop:lower-bound-1}
  For all $N$ and $K<N/5$,
  \begin{equation}\label{eq:lower-bound-1}
    \UP
    \left(
      \Loss_N(A) - \Loss_N
      \ge
      \sqrt{K N}
    \right)
    \ge
    \epsilon,
  \end{equation}
  where $A$ is the Bayes optimal strategy
  and $\epsilon$ is a universal positive constant.
\end{proposition}

The lower bound $\sqrt{K N}$ in \eqref{eq:lower-bound-1}
matches the upper bound in \eqref{eq:optimal}
(Theorem~\ref{thm:optimal})
as far as $K$ and $N$ are concerned.
(The result in \eqref{eq:optimal} is best interpreted as an upper bound,
despite the inequality ``$\ge$'';
this can be seen from its restatement in the form~\eqref{eq:upper}.)

See Appendix~\ref{app:martingale-SLLN}
for related results
(Propositions~\ref{prop:LLN} and~\ref{prop:anti-LLN})
in measure-theoretic probability.

Proposition~\ref{prop:lower-bound-1} only concerns the optimality
of the upper bound in \eqref{eq:optimal} in $K$ and $N$,
but the next proposition shows that it is also close to being optimal in~$\epsilon$.
In this proposition we use a slightly different definition of $\Loss_N$:
now, unlike in \eqref{eq:Loss-mod},
we sum the losses of all decisions, including those of $d_{N-K+2},\dots,d_N$
(they will be defined in a very natural way).

\begin{proposition}\label{prop:lower-bound-2}
  Suppose that $N$ and $K$ are such that $\sqrt{N/K}$ is an even integer.
  Then the Bayes optimal strategy $A$ satisfies,
  for any $\epsilon>0$ such that $\sqrt{\ln\frac{1}{\epsilon}}$ is integer,
  \begin{equation}\label{eq:lower-bound-2}
    \UP
    \left(
      \Loss_N(A) - \Loss_N
      \ge
      \sqrt{K N \ln\frac{1}{\epsilon}}
    \right)
    \ge
    \epsilon^4/15
  \end{equation}
  provided
  \begin{equation}\label{eq:condition}
    \sqrt{K N \ln\frac{1}{\epsilon}}
    \le
    N/4.
  \end{equation}
\end{proposition}

The condition~\eqref{eq:condition} is mild in this context;
without it, the bound \eqref{eq:lower-bound-2} appears useless.
The substitution $\epsilon:=\epsilon^4/15$ in Proposition~\ref{prop:lower-bound-2}
gives the following corollary,
which shows that the upper bound in \eqref{eq:optimal} is optimal
if we ignore additive and multiplicative constants in the ``regret term''
\[
  2 \sqrt{K N \ln\frac{1}{\epsilon}}.
\]

\begin{corollary}
  Under the conditions of Proposition~\ref{prop:lower-bound-2},
  \begin{equation}\label{eq:corollary}
    \UP
    \left(
      \Loss_N(A) - \Loss_N
      \ge
      \frac12
      \sqrt{K N \ln\frac{1}{15\epsilon}}
    \right)
    \ge
    \epsilon
  \end{equation}
  provided the term $\sqrt{\strut\dots}$ does not exceed $N/2$.
\end{corollary}

For proofs of Propositions~\ref{prop:lower-bound-1} and~\ref{prop:lower-bound-2},
see Sects~\ref{subsec:lower-bound-1} and~\ref{subsec:lower-bound-2},
respectively.

\section{Conclusion}
\label{sec:conclusion}

This paper has scratched the surface of the diachronic picture of realistic Bayesian forecasting
not based on Bayesian conditioning.
We discussed ways of testing such forecasts based on betting
and their applications to Bayesian decision making.

Obvious directions of further research include, e.g.,
considering an infinite time horizon and more general observation spaces $\mathbf{Y}$.
Another direction is to generalize our basic forecasting protocol:
instead of assuming that the forecaster observes a new outcome $y_n$ at each step,
we could consider cases where beliefs are revised
(perhaps because new information arrives from outside the protocol)
without new outcomes becoming known;
for a first step in this direction,
see Appendix~\ref{app:radical}.

\subsection*{Acknowledgments}

My research has been partially supported by Mitie.
Many thanks to Philip Dawid for advice on literature and useful discussions
and to Ilia Nouretdinov for his input.
Comments by the participants in the International Seminar on Selective Inference
are gratefully appreciated.

\appendix
\section{Proofs}
\label{app:proofs}

\subsection{Proof of Proposition~\ref{prop:coherence-useless}}
\label{subsec:coherence-useless}

Let us fix such sequences of probability measures
$P_n\in\mathfrak{P}(\mathbf{Y}^{N-n+1})$
and outcomes $y_n\in\mathbf{Y}$.
As a first step,
define $\Omega$ as $\mathbf{Y}^N$,
$P$ as $P_1$,
$Y_n(\omega)$ as the $n$th element $\omega_n$ of $\omega\in\Omega$,
and let the $\sigma$-algebra $\FFF_n$ be generated by $Y_1,\dots,Y_n$.

Next modify the finite probability space $(\Omega,P)$ and filtration $(\FFF_n)$
as follows.
Split each sample point $y_1\omega_2\dots\omega_N$
that starts from $y_1$ into two sample points,
$y'_1\omega_2\dots\omega_N$ and $y''_1\omega_2\dots\omega_N$,
and make the sets $\{y'_1\}\times\Omega^{N-1}$ and $\{y''_1\}\times\Omega^{N-1}$
$\FFF_n$-measurable for $n\ge 1$.
Split the old value $c:=P_1(\{y_1\}\times\Omega^{N-1})$
into $P(\{y'_1\}\times\Omega^{N-1}):=\epsilon c$, for a sufficiently small $\epsilon>0$,
and $P(\{y''_1\}\times\Omega^{N-1}):=(1-\epsilon)c$.
Without changing $P(\{\omega_1\dots\omega_N\})$ for $\omega_1\notin\{y'_1,y''_1\}$,
set
\[
  \frac{P(\{y'_1\omega_2\dots\omega_N\})}{P(\{y'_1\}\times\Omega^{N-1})}
  :=
  P_2(\{\omega_2\dots\omega_N\}),
  \quad
  \omega_2,\dots,\omega_N\in\Omega,
\]
and define 
\[
  \frac{P(\{y''_1\omega_2\dots\omega_N\})}{P(\{y'_1\}\times\Omega^{N-1})},
  \quad
  \omega_2,\dots,\omega_N\in\Omega,
\]
in such a way that we have an agreement with $P_1$:
\begin{equation*}
  \forall\omega_2,\dots,\omega_N\in\Omega:
  \frac{P(\{y'_1\omega_2\dots\omega_N,y''_1\omega_2\dots\omega_N\})}{P(\{y'_1,y''_1\}\times\Omega^{N-1})}
  =
  \frac{P_1(\{y_1\omega_2\dots\omega_N\})}{P_1(\{y_1\}\times\Omega^{N-1})};
\end{equation*}
this is possible for a sufficiently small $\epsilon>0$.

Apply the same procedure to the probability subspace of $(\Omega,P)$
consisting of the sample points $y'_1\omega_2\dots\omega_N$,
thereby splitting $y_2$ into $y'_2$ and $y''_2$.
Continue by splitting $y_3$, $y_4$, etc.

\subsection{Proof of Theorem~\ref{thm:validity} and Proposition~\ref{prop:validity}}
\label{subsec:validity}

We start from Theorem~\ref{thm:validity}.
Let us consider in detail only the first step in Protocol~\ref{prot:joint-test-final},
when we move from prediction $P:=P_1$ to prediction $Q:=P_2$ (the rest will be easy).
We regard $P$ as fixed (so that our argument is conditional on $P$)
and use the notation $P(y)$, where $y\in\mathbf{Y}$,
and $P(x\mid y)$, where $y\in\mathbf{Y}$ and $x\in\mathbf{Y}^{N-1}$,
as usual.
We also use $Q_r(x\mid y)$ for the various $Q=P_2$ possible
after observing $y$ as the first observation $y_1$
(and we assume that $Q_r$ are all different).
We will be able to apply the standard duality theorem
since $r$ ranges over a finite set;
remember that we consider a finite probability space.
It will be convenient to refer to $S_1$ as the \emph{first value}
of a game-theoretic or measure-theoretic martingale $(S_n)_{n=0}^N$.

In Sect.~\ref{subsec:GTP-valid} we saw that every game-theoretic test martingale
is a visible measure-theoretic test martingale,
and this implies that every game-theoretic test supermartingale
is a visible measure-theoretic test supermartingale;
therefore, we will be only interested in the opposite direction.
Let $S_{y,r}$ be the first value of a visible measure-theoretic test supermartingale $(S_n)_{n=0}^N$;
i.e., $S_{y,r}$ is the first value $S_1$ when we observe $y_1=y$ and $P_2=Q_r$.
Our goal is to show that $S_{y,r}$ is the first value of a game-theoretic test supermartingale.

The primary (measure-theoretic) linear programming problem
involves variables $X_{y,r}\ge0$ subject to the constraints
\begin{equation}
  \sum_r X_{y,r} = 1
  \notag
\end{equation}
for all $y$ and
\begin{equation}\label{eq:C-i-2}
  \sum_r X_{y,r} Q_r(x\mid y) = P(x\mid y)
\end{equation}
for all $y$ and $x$.
The interpretation is that $X_{y,r}$ is the conditional probability of $Q_r$
after observing $y$.
The relevant optimization problem is
\begin{equation}\label{eq:primary-objective}
  \sum_y
  P(y)
  \sum_r
  S_{y,r} X_{y,r}
  \to
  \max.
\end{equation}
By the choice of $S$, the max value is at most $1$.

The dual (game-theoretic) problem is
\begin{equation}\label{eq:dual-objective}
  \sum_{y,x}
  P(x\mid y)
  Y_{y,x}
  +
  \sum_{y}
  Y_{y}
  \to
  \min
\end{equation}
subject to
\begin{equation}\label{eq:Q-j}
  \sum_x
  Q_r(x\mid y)
  Y_{y,x}
  +
  Y_{y}
  \ge
  P(y)
  S_{y,r}
\end{equation}
for all $y,r$.
(The recipe for stating the dual problem
given in \citealt[Sect.~6.2]{Matousek/Gartner:2007},
is particularly convenient in this context.)
The dual variables $Y_{y,x}$ and $Y_y$ are unconstrained.
Rewriting \eqref{eq:dual-objective} and \eqref{eq:Q-j} as
\begin{align*}
  \sum_{y,x}
  P(x\mid y)
  (Y_{y,x}+Y_{y})
  &\to
  \min\\
  \sum_x
  Q_r(x\mid y)
  (Y_{y,x}+Y_{y})
  &\ge
  P(y)
  S_{y,r},
\end{align*}
respectively,
we can see that the optimization problem \eqref{eq:dual-objective}--\eqref{eq:Q-j}
is equivalent to
\begin{equation}\label{eq:dual-1}
  \sum_{y,x}
  P(x\mid y)
  Y_{y,x}
  \to
  \min
  \quad
  \text{subject to}
  \quad
  \sum_x
  Q_r(x\mid y)
  Y_{y,x}
  \ge
  P(y)
  S_{y,r}.
\end{equation}
Replacing the variables $Y_{y,x}$ with new variables $Z_{y,x}$
defined by $Y_{y,x}=P(y)Z_{y,x}$,
we rewrite the optimization problem \eqref{eq:dual-1} as
\begin{equation}\label{eq:dual-2}
  \sum_{y,x}
  P(y x)
  Z_{y,x}
  \to
  \min
  \quad
  \text{subject to}
  \quad
  \sum_x
  Q_r(x\mid y)
  Z_{y,x}
  \ge
  S_{y,r},
\end{equation}
with the same value, at most 1.
Any solution to the optimization problem \eqref{eq:dual-2} achieves our goal:
setting $f_1(y x):=Z_{y,x}$,
our portfolio of tickets will have the total final price at least $S_1$
while their total initial price will be at most 1.
To complete the proof of Theorem~\ref{thm:validity},
we need to apply the same argument conditionally on the first $n$ observations
$y_1,\dots,y_n$
for $n=1,\dots,N-1$.

Before proving Proposition~\ref{prop:validity}
let us make a short detour and check that
every visible measure-theoretic test supermartingale
is dominated by a visible measure-theoretic test martingale.
First we make $S=(S_{y,r})$ admissible
replacing each $S_{y,r}$ by the left-hand side of the constraint
in \eqref{eq:dual-2}.
The expression being maximized in \eqref{eq:primary-objective} becomes
\begin{align*}
  \sum_y
  P(y)
  \sum_r
  X_{y,r} S_{y,r}
  &=
  \sum_y
  P(y)
  \sum_r
  X_{y,r}
  \sum_x
  Q_r(x\mid y)
  Z_{y,x}\\
  &=
  \sum_y
  P(y)
  \sum_x
  Z_{y,x}
  P(x\mid y)
  =
  \sum_{y,x}
  P(y x)
  Z_{y,x},
\end{align*}
where the second equality uses \eqref{eq:C-i-2}.
The last expression is very natural,
and does not depend at all on the primary variables $X_{y,r}$,
which shows that $S_{y,r}$ is the first value of a visible measure-theoretic test martingale
except that its initial value can be below 1
(in which case it can be scaled up to make
its initial value equal to 1).

Finally, if $(S_{y,r})$ is the first value of a visible measure-theoretic test martingale,
it will coincide with the first value of a game-theoretic test supermartingale,
which will be the first value of a game-theoretic test martingale
(otherwise we could increase this game-theoretic test supermartingale
to obtain a visible measure-theoretic martingale
whose first value would strictly dominate the first value
of the original visible measure-theoretic test martingale,
which is impossible).
Repeating this argument for $y_2,\dots,y_N$ completes the proof of Proposition~\ref{prop:validity}.

\subsection{Game-theoretic probability}
\label{subsec:GTP}

In the proof of Theorem~\ref{thm:optimal} in Sect.~\ref{subsec:optimal}
we will need some basic definitions and results
in game-theoretic probability given in this subsection;
see \citet{Shafer/Vovk:2019} for further information.
We will let $\E_n$ denote the game-theoretic expectation
(to be defined momentarily)
at the point in Protocol~\ref{prot:DM-4}
right after Decision Maker announcing her move $d_n$
(let us call this point the \emph{checkpoint}).
In our current context $\E_n$ can be defined as follows.
If $f=f(y_n\dots y_{(n+K-1)\wedge N})$ is a function
of the $K$ consecutive moves by Reality starting from $y_n$
(and ending with $y_N$ if $n+K-1\ge N$),
\[
  \E_n f
  :=
  \sum_{x\in\mathbf{Y}^{K\wedge(N-n+1)}}
  f(x) P_n(x).
\]
More generally, if $f$ depends on other future moves (by Reality and other players),
$\E_n f$ is the initial capital (if it exists) starting from which Sceptic can attain
exactly the final capital of $f$ at the end of step $N$.
If $f$ also depends on the moves preceding the step $n$ checkpoint,
$\E_n f$ is found separately for each set of these preceding moves.

The \emph{game-theoretic sample space} $\Omega$
consists of all possible sequences of moves
\[
  \omega
  :=
  \left(
    \lambda_1,P_1,d_1,y_1,
    \dots
    \lambda_N,P_N,d_N,y_N
  \right)
\]
by non-Sceptic players in Protocol~\ref{prot:DM-4}.
A \emph{nonnegative variable} $X$ is a function
$X:\Omega\to[0,\infty)$.
The \emph{upper expectation} of $X$ is defined as 
\begin{equation*}
  \UE(X)
  :=
  \inf
  \left\{
    \alpha>0
    \mid
    \exists\text{ strategy for Sceptic }
    \forall\omega\in\Omega:
    \alpha\K_N(\omega)\ge X(\omega)
  \right\},
\end{equation*}
where $\omega$ are the non-Sceptic player's moves
and $\K_N$ is regarded as function of $\omega$.
In words, $\UE(X)$ is the smallest (in the sense of $\inf$) initial capital
that Sceptic can turn into $X(\omega)$ or more.
An \emph{event} is a set $E\subseteq\Omega$.
The \emph{upper probability} $\UP(E)$ of an event $E$
is defined to be $\UE(1_E)$.

\begin{lemma}\label{lem:representation}
  For any bounded nonnegative variable $X$,
  \begin{equation}\label{eq:representation}
    \UE(X)
    \le
    \int_0^{\infty}
    \UP(X\ge u)
    \d u.
  \end{equation}
\end{lemma}

\begin{proof}
  Set $f(u):=\UP(X\ge u)$;
  then $f:[0,\infty)\to[0,1]$ is a decreasing function.
  Replace the $\infty$ in \eqref{eq:representation}
  by $C$ for some upper bound $C$ for $X$.
  For each $k=0,\dots,\lceil C/\epsilon\rceil$,
  fix a strategy for Sceptic that turns $f(k\epsilon)+\epsilon$
  into $1_{\{X\ge k\epsilon\}}$ or more.
  Multiplying this strategy and its initial capital by $\epsilon$
  and then summing over the $k$,
  we obtain a strategy that turns
  \begin{equation}\label{eq:initial}
    \sum_{k=0}^{\lceil C/\epsilon\rceil}
    \epsilon(f(k\epsilon)+\epsilon)
  \end{equation}
  into at least
  \[
    \sum_{k=0}^{\lceil C/\epsilon\rceil}
    \epsilon 1_{\{X(\omega)\ge k\epsilon\}}
    \ge
    X(\omega).
  \]
  It remains to notice that \eqref{eq:initial}
  tends to $\int_0^C f(u) \d u$ as $\epsilon\to0$.
\end{proof}

\subsection{Proof of Theorem~\ref{thm:optimal}}
\label{subsec:optimal}

This subsection uses the definitions and results from game-theoretic probability
given in Sect.~\ref{subsec:GTP}.
The reader familiar with measure-theoretic probability
who encounters game-theoretic probability for the first time
might prefer to read Appendix~\ref{app:martingale-SLLN} first
as a gentle introduction to the rest of this section.

We will also need the following lemma,
which is widely used in robust risk aggregation
(and our use of this lemma will mimic its uses in robust risk aggregation).

\begin{lemma}\label{lem:EP}
  For any $C>0$, any $\alpha\in(0,C/K)$,
  and any $x_1,\dots,x_K\in\R$,
  \begin{equation}\label{eq:EP}
    \sum_{k=1}^K
    g(x_k)
    \ge
    1_{\{\sum_{k=1}^K x_k \ge C\}},
  \end{equation}
  where $g$ is the continuous function
  \begin{equation}\label{eq:g}
    g(x)
    :=
    \begin{cases}
      0 & \text{if $x<C/K-\alpha$}\\
      \frac{x-(C/K-\alpha)}{K\alpha} & \text{if $C/K-\alpha\le x\le C/K+(K-1)\alpha$}\\
      1 & \text{if $x>C/K+(K-1)\alpha$}.
    \end{cases}
  \end{equation}
\end{lemma}

\begin{proof}
  We argue indirectly.
  Suppose there is a set of numbers $x_1,\dots,x_K$
  for which \eqref{eq:EP} holds with ``$<$'' in place of ``$\ge$'',
  and let us fix such a set.
  If $x_i<C/K-\alpha$ and $x_j>C/K+(K-1)\alpha$,
  we can replace $x_i$ by $x_i+t$ and $x_j$ by $x_j-t$,
  where $t>0$ is the smallest number such that
  $x_i+t=C/K-\alpha$ or $x_j-t=C/K+(K-1)\alpha$;
  therefore, we can assume, without loss of generality, that there is no such pair $(i,j)$.
  In this case, $x_k\le C/K+(K-1)\alpha$ for all $k$,
  but perhaps $x_j<C/K-\alpha$ for some $j$.
  It remains to apply Jensen's inequality
  to the convex (and increasing) function $g|_{(-\infty,C/K+(K-1)\alpha]}$:
  as the average of $x_k$ is at least $C/K$,
  the average of $g(x_k)$ is at least $g(C/K)=1/K$.
\end{proof}

\noindent
See the proof of Theorem~4.2 in \citet{Embrechts/Puccetti:2006}
for another proof of Lemma~\ref{lem:EP},
and see Appendix~\ref{app:martingale-SLLN} for further information
about robust risk aggregation.

Set $Q:=\lfloor N/K\rfloor$.
Let us first assume that $N=Q K+K-1$;
later we will get rid of this assumption
(it will be easy as $N=Q K+K-1$ is, in a sense, the worst case).

To get a handle on the difference $\Loss_N(A)-\Loss_N$ in Protocol~\ref{prot:DM-4},
we first consider its increment
\begin{equation}\label{eq:increment}
  \lambda(d_i^A,y_i\dots y_{i+K-1})
  -
  \lambda(d_i,y_i\dots y_{i+K-1})
\end{equation}
on step $i\le N-K+1$,
where $d_i^A$ is the prediction output
by the strategy $A$ defined by \eqref{eq:AA}.
By the choice of $d_i^A$,
the difference \eqref{eq:increment} is a supermartingale difference,
meaning that its $\E_i$ expectation is nonpositive.
Namely,
\begin{multline*}
  \E_i
  \bigl(
    \lambda(d_i^A,y_i\dots y_{i+K-1})
    -
    \lambda(d_i,y_i\dots y_{i+K-1})
  \bigr)\\
  =
  \sum_{x\in\mathbf{Y}^{K}}
  \left(
    \lambda(d^A_i,x)
    -
    \lambda(d_i,x)
  \right)
  P_i(x)
  \le
  0.
\end{multline*}

For each $k\in\{1,\dots,K\}$,
we consider the process
\begin{multline}\label{eq:L-1}
  L^k_n
  =
  \E_n
  \sum_{i\in\{k,k+K,\dots,k+(Q-1)K\}}
  \biggl(
    \lambda(d_i^A,y_i\dots y_{i+K-1})
    -
    \lambda(d_i,y_i\dots y_{i+K-1})\\
    +
    \sum_{x\in\mathbf{Y}^{K}}
    \left(
      \lambda(d_i,x)
      -
      \lambda(d_i^A,x)
    \right)
    P_i(x)
  \biggr);
\end{multline}
including only every $K$th step in the sum simplifies the analysis
and, more importantly, makes the result stronger
(cf.\ Remark~\ref{rem:strong-law}).
This process starts from zero,
and it is a game-theoretic martingale
(namely, $L^k_n=\K_{n-1}$ for some strategy for Sceptic
in the modification of Protocol~\ref{prot:DM-4}
replacing $\K_n:=1$ by $\K_n:=0$ and allowing $\K$ to become negative),
as the following explicit expression shows:
\begin{multline}\label{eq:L-2}
  L^k_n
  :=
  \sum_{i\in\{k,k+K,\dots,k+(q-1)K\}}
  \biggl(
    \lambda(d_i^A,y_i\dots y_{i+K-1})
    -
    \lambda(d_i,y_i\dots y_{i+K-1})\\
    +
    \sum_{x\in\mathbf{Y}^K}
    \left(
      \lambda(d_i,x)
      -
      \lambda(d_i^A,x)
    \right)
    P_i(x)
  \biggr)\\
  +
  \sum_{x\in\mathbf{Y}^{K-j}}
  \bigl(
    \lambda(d_{k+q K}^A,y_{k+q K}\dots y_{n-1}x)
    -
    \lambda(d_{k+q K},y_{k+q K}\dots y_{n-1}x)
  \bigr)
  P_n(x)\\
  +
  \sum_{x\in\mathbf{Y}^{K}}
  \left(
    \lambda(d_{k+q K},x)
    -
    \lambda(d_{k+q K}^A,x)
  \right)
  P_{k+q K}(x)
\end{multline}
where $q$ and $j\in\{0,\dots,K-1\}$ are the integers
from the representation $n=k+q K+j$,
and we are only interested in $n\le Q K$.
The first sum (i.e., the sum $\sum_{i\in\{k,k+K,\dots,k+(q-1)K\}}$) in \eqref{eq:L-2}
includes the terms \eqref{eq:increment} (for $i\equiv k \pmod{K}$)
that are determined by the checkpoint on step $n$.
The rest of the expression in \eqref{eq:L-2}
accounts for the term \eqref{eq:increment} that is partially determined,
which corresponds to $i=k+q K$.
And we do not have terms corresponding to $i>k+q K$
since at the checkpoint on step $n$ the expectation
of the expression in the outer parentheses in \eqref{eq:L-1}
is still 0 for such~$i$.

To check that \eqref{eq:L-2} is indeed a game-theoretic martingale,
it suffices to notice that
\begin{equation*}
  L^k_{n} - L^k_{n-1}
  =
  \sum_{x\in\mathbf{Y}^{K-j}}
  f_{n-1}(y_{n-1}x)
  P_n(x)
  -
  \sum_{x\in\mathbf{Y}^{K-j+1}}
  f_{n-1}(x)
  P_{n-1}(x),
\end{equation*}
where
\[
  f_{n-1}(x)
  :=
  \lambda(d_{k+q K}^A,y_{k+q K}\dots y_{n-2}x)
  -
  \lambda(d_{k+q K},y_{k+q K}\dots y_{n-2}x),
\]
has the same form as the capital increment in \eqref{eq:joint-K-P-4}.
This assumes that $n$ is not one of the borderline values $k+q K$,
which case should be considered separately.

If we only consider the values of the game-theoretic martingale \eqref{eq:L-2}
at steps $k+q K$, $q=0,1,\dots,Q$,
\begin{multline}\label{eq:L-3}
  L^k_{k+q K}
  :=
  \sum_{i\in\{k,k+K,\dots,k+(q-1)K\}}
  \biggl(
    \lambda(d_i^A,y_i\dots y_{i+K-1})
    -
    \lambda(d_i,y_i\dots y_{i+K-1})\\
    +
    \sum_{x\in\mathbf{Y}^K}
    \left(
      \lambda(d_i,x)
      -
      \lambda(d_i^A,x)
    \right)
    P_i(x)
  \biggr),
  \quad
  q=0,1,\dots,Q,
\end{multline}
its increments will be bounded by 2 in absolute value,
and we can apply the game-theoretic Hoeffding inequality
\citep[Corollary 3.8 for Protocol~3.5]{Shafer/Vovk:2019}
to it.
However, a tighter inequality is obtained
when we apply the one-sided version
of the game-theoretic Hoeffding inequality
\citep[Corollary 3.8 for Protocol~3.7]{Shafer/Vovk:2019}
to the process \eqref{eq:L-3}
with the sum over $x\in\mathbf{Y}^K$ removed.
This process is a game-theoretic supermartingale
whose increments are bounded by 1 in absolute value,
and the one-sided Hoeffding inequality gives
\begin{equation}\label{eq:Hoeffding-gives}
  \UP
  \left(
    X_k
    \ge
    U
  \right)
  \le
  \exp
  \left(
    -\frac{U^2}{2Q}
  \right)
  \le
  \exp
  \left(
    -U^2\frac{K}{2N}
  \right),
\end{equation}
where $U\ge0$ and
\[
  X_k
  :=
  \sum_{i\in\{k,k+K,\dots,k+(Q-1)K\}}
  \left(
    \lambda(d_i^A,y_i\dots y_{i+K-1})
    -
    \lambda(d_i,y_i\dots y_{i+K-1})
  \right);
\]
we assume that the game-theoretic supermartingale is constant after $k+(Q-1)K$
(the last $i$ in the range of summation in \eqref{eq:L-3}).

Applying~\eqref{eq:Hoeffding-gives}
and Lemmas~\ref{lem:representation} and~\ref{lem:EP}
(see below for details)
gives, for any $C>0$,
\begin{align}
  \UP&(X_1+\dots+X_K\ge C)
  \le
  \sum_{k=1}^K
  \UE(g(X_k))
  \le
  \sum_{k=1}^K
  \int_0^{\infty}
  \UP(g(X_k)\ge u)
  \d u\label{eq:1and2}\\
  &\le
  \sum_{k=1}^K
  \int_0^{\infty}
  \UP
  \left(
    X_k\ge\gamma\frac{C}{K}+(1-\gamma)C u
  \right)
  \d u\label{eq:3}\\
  &\le
  \sum_{k=1}^K
  \int_0^{\infty}
  \exp
  \left(
   -\left(
      \gamma\frac{C}{K}+(1-\gamma)C u
    \right)^2
    \frac{K}{2N}
  \right)
  \d u\label{eq:4}\\
  &=
  \frac{\sqrt{K N}}{(1-\gamma)C}
  \int_{\frac{\gamma C}{\sqrt{K N}}}^{\infty}
  \exp
  \left(
    -v^2/2
  \right)
  \d v\label{eq:5}\\
  &=
  \frac{\sqrt{K N}}{(1-\gamma)C}
  \sqrt{2\pi}
  \bar\Phi
  \left(
    \frac{\gamma C}{\sqrt{K N}}
  \right)
  <
  \frac{K N}{\gamma(1-\gamma)C^2}
  \exp
  \left(
    -\frac{\gamma^2 C^2}{2 K N}
  \right).\label{eq:6and7}
\end{align}
The first and second inequalities in \eqref{eq:1and2}
follow from Lemmas~\ref{lem:EP} and~\ref{lem:representation},
respectively.
The inequality~\eqref{eq:3} follows from the definition of $g$ in \eqref{eq:g}
with $\alpha:=(1-\gamma)C/K$.
Indeed, we can assume, without loss of generality, $u>0$,
and then $g(X)\ge u$ implies
\[
  \frac{X-(C/K-\alpha)}{K\alpha}
  \ge
  u,
\]
which is equivalent to
\[
  X
  \ge
  \gamma\frac{C}{K}+(1-\gamma)C u.
\]
The inequality~\eqref{eq:4} follows
from Hoeffding's inequality~\eqref{eq:Hoeffding-gives}.
The equality~\eqref{eq:5} follows by the substitution
\[
  v
  :=
  \frac{\gamma C}{\sqrt{K N}}
  +
  (1-\gamma)C
  \sqrt{\frac{K}{N}}
  u.
\]
The equality in~\eqref{eq:6and7} introduces the notation
$\bar\Phi:=1-\Phi$
for the survival function of the standard Gaussian distribution.
And the last inequality in the chain
follows by applying the standard upper bound \citep[Lemma~VII.1.2]{Feller:1968} on $\bar\Phi$.

We can rewrite the inequality between the extreme terms
in the chain~\eqref{eq:1and2}--\eqref{eq:6and7} as
\begin{equation}\label{eq:extreme}
  \UP
  \left(
    \Loss_N(A) - \Loss_N
    \ge
    C
  \right)
  \le
  \frac{K N}{\gamma(1-\gamma)C^2}
  \exp
  \left(
    -\frac{\gamma^2 C^2}{2 K N}
  \right).
\end{equation}
Comparing this with \eqref{eq:optimal},
we can see that we need to solve the inequality
\begin{equation}\label{eq:inequality-1}
  \frac{K N}{\gamma(1-\gamma)C^2}
  \exp
  \left(
    -\frac{\gamma^2 C^2}{2 K N}
  \right)
  \le
  \epsilon.
\end{equation}
Ignoring the part before the $\exp$ and replacing ``$\le$'' by ``$=$'',
we obtain the solution
\[
  C
  =
  \frac{\sqrt{2 K N \ln\frac{1}{\epsilon}}}{\gamma},
\]
which motivates the substitution
\begin{equation}\label{eq:C}
  C
  :=
  \frac{\sqrt{2 K N \ln\frac{1}{\epsilon}x}}{\gamma}
\end{equation}
in \eqref{eq:inequality-1}.
After this substitution,
\eqref{eq:inequality-1} simplifies to
\begin{equation}\label{eq:inequality-2}
  \epsilon^{x-1}
  \le
  2 \frac{1-\gamma}{\gamma} x \ln\frac{1}{\epsilon}.
\end{equation}
Setting $x:=2\gamma^2$ in \eqref{eq:C}
gives an expression that matches the corresponding expression in \eqref{eq:optimal}.

The condition $x>1$
(required for \eqref{eq:inequality-2} to hold as $\epsilon\to0$)
narrows down the range of $\gamma$ from $\gamma\in(0,1)$ to $\gamma\in(2^{-1/2},1)$.
Setting, e.g., $\gamma:=0.8$
ensures that \eqref{eq:inequality-2} holds for all $\epsilon\in(0,0.32)$.

It remains to consider the case $N<Q K+K-1$.
If the final value of $L^k_{k+q K}$ (corresponding to $q=Q$)
is undefined (because $k+q K+K-1>N$),
we set it equal to its previous value (for $q=Q-1$).

This completes the proof of Theorem~\ref{thm:optimal}.
Inequality~\eqref{eq:no-restriction} follows from \eqref{eq:extreme}
with $\gamma:=1/\sqrt{2}$.

\subsection{Proof of Proposition~\ref{prop:lower-bound-1}}
\label{subsec:lower-bound-1}

Similarly to \eqref{eq:binary-loss},
let us set $\mathbf{D}:=\mathbf{Y}:=\{0,1\}$ and
\begin{equation}\label{eq:lambda}
  \lambda_n(d_n,y_n\dots y_{n+K-1})
  :=
  \begin{cases}
    1 & \text{if $d_n\ne y_{\lceil n/K\rceil K}$}\\
    0 & \text{otherwise}.
  \end{cases}
\end{equation}
We are only interested in $n\le N-K+1$
(see \eqref{eq:Loss-mod}),
which implies $n+K-1\le N$ and $\lceil n/K\rceil K\le N$;
therefore, \eqref{eq:lambda} is well-defined.
Now the true probability measure $P$ is such that
$y_n=1$ with probability $1/2$ independently for different $n$
(and now we will rely on our tie-breaking convention).
As in Sect.~\ref{subsec:essential}, the players comply with $P$.
Let $B$ be the decision strategy that always outputs $1$;
notice that $A$ always outputs $0$
(assuming that the linear order on $\mathbf{D}$ is $0<1$).
It suffices to prove \eqref{eq:lower-bound-1}
(and later  \eqref{eq:lower-bound-2})
with $\Loss_N$ replaced by $\Loss_N(B)$,
and this is what we will do.

The $N$ steps are now split into $\lceil N/K\rceil$ blocks of $K$ steps
(except, possibly, the last block),
$n\in\{1,\dots,K\}$,
$n\in\{K+1,\dots,2K\}$,
etc.
Within each block,
$A$ suffers the same loss at each step,
and $B$ suffers the same loss at each step.
By the central limit theorem,
the probability is at least $\epsilon$ (a universal positive constant)
that $A$ performs worse than $B$
in at least $\sqrt{N/K}+1$ more blocks than vice versa.
In such cases
\[
  \Loss_N(A) - \Loss_N(B)
  \ge
  K\sqrt{N/K}
  =
  \sqrt{K N}.
\]
This gives \eqref{eq:lower-bound-1} with $P$ in place of $\UP$.
By Ville's inequality,
we can replace the probability measure $P$
by the upper game-theoretic probability $\UP$.

\subsection{Proof of Proposition~\ref{prop:lower-bound-2}}
\label{subsec:lower-bound-2}

We will obtain Proposition~\ref{prop:lower-bound-2}
by applying a lower bound for large deviations
in the form of \citet[Proposition 7.3.2]{Matousek/Vondrak:2008}
to the argument of the previous subsection.
As mentioned before the statement of the proposition,
now we define $\Loss_N$ by summing the losses over all steps $n=1,\dots,N$,
as in \eqref{eq:Loss}.
The loss at each step $n$ is still given by the right-hand side of \eqref{eq:lambda}.
This is the derivation (see below for some explanations):
\begin{align}
  \UP&
  \left(
    \Loss_N(A) - \Loss_N(B)
    \ge
    \sqrt{K N \ln\frac{1}{\epsilon}}
  \right)\notag\\
  &=
  \UP
  \left(
    \frac1K\Loss_N(A)
    \ge
    \frac{N}{2K}
    +
    \frac12
    \sqrt{\frac{N}{K} \ln\frac{1}{\epsilon}}
  \right)\label{eq:c1}\\
  &=
  \UP
  \left(
    X
    \ge
    \frac{n}{2}
    +
    t
  \right)
  \ge
  \frac{1}{15}
  \exp
  \left(
    -16 t^2 / n
  \right)
  =
  \epsilon^4/15.\label{eq:c2}
\end{align}
The first equality, \eqref{eq:c1},
follows from
\begin{equation}\label{eq:sum}
  \Loss_N(A) + \Loss_N(B) = N
\end{equation}
(which allows us to eliminate $\Loss_N(B)$)
and obvious transformations.
The first equality in~\eqref{eq:c2}
introduces the notation
\[
  X := \frac1K\Loss_N(A),
  \quad
  n := \frac{N}{K},
  \quad
  t
  :=
  \frac12
  \sqrt{\frac{N}{K} \ln\frac{1}{\epsilon}},
\]
which is the notation
used in \citet[Proposition 7.3.2]{Matousek/Vondrak:2008}.
The inequality ``$\ge$'' in~\eqref{eq:c2}
is identical to \citet[Proposition 7.3.2]{Matousek/Vondrak:2008}.

\citet{Matousek/Vondrak:2008} have two conditions
in their Proposition 7.3.2:
$t\le n/8$, which becomes~\eqref{eq:condition},
and $t$ being an integer,
which we strengthen to $\sqrt{N/K}$ being an even integer
and $\sqrt{\ln\frac{1}{\epsilon}}$ being an integer.

\begin{remark}
  \citet[Lemma~3]{Kunsch/Rudolf:2019}
  slightly improve the constants in \citet[Proposition 7.3.2]{Matousek/Vondrak:2008},
  and using their result we can improve the bound $\epsilon^4/15$
  in \eqref{eq:c2} to $\epsilon^3/5$.
  This allows us to rewrite \eqref{eq:corollary} in the form
  \begin{equation}\label{eq:KR}
    \UP
    \left(
      \Loss_N(A) - \Loss_N(B)
      \ge
      \sqrt{\frac13 K N \ln\frac{1}{5\epsilon}}
    \right)
    \ge
    \epsilon.
  \end{equation}
\end{remark}

\begin{remark}
  Let us check informally
  what the optimal counterparts of the constants $1/3$ and $5$ in \eqref{eq:KR}
  would be in the domain of applicability of the central limit theorem.
  We have for the probability measure $P$ of Sect.~\ref{subsec:lower-bound-1}:
  \begin{equation*}
    P
    \left(
      \Loss_N(A) - \Loss_N(B)
      \ge
      \bar\Phi^{-1}(\epsilon)
      \sqrt{K N}
    \right)
    \approx
    \epsilon,
  \end{equation*}
  where $\bar\Phi^{-1}(\epsilon)$ is the upper $\epsilon$-quantile
  of the standard Gaussian distribution.
  This follows from the variance of
  \[
    \Loss_N(A) - \Loss_N(B)
    =
    2\Loss_N(A) - N
  \]
  (cf.\ \eqref{eq:sum}) being approximately $K N$.
  This gives the ideal approximate equality
  \begin{equation*}
    P
    \left(
      \Loss_N(A) - \Loss_N(B)
      \ge
      \sqrt{2 K N \ln\frac{1}{\epsilon}}
    \right)
    \approx
    \epsilon
  \end{equation*}
  in place of \eqref{eq:KR}.
  It is interesting that this is exactly what we get from \eqref{eq:C}
  when we make $\gamma\approx1$ and $x\approx1$
  (it is clear that we can make $\gamma\in(0,1)$ and $x>1$
  as close to 1 as we want
  at the price of restricting $\epsilon$ to a narrower range $(0,\epsilon^*)$).
\end{remark}

\section{Protocols in terms of ideal futures}
\label{app:ideal}

Testing by betting in general and game-theoretic probability in particular
can be interpreted in terms of trading in a financial market,
and in this appendix I will make this interpretation explicit.
Now we complement the basic forecasting protocol
with an idealized market allowing Sceptic to trade in futures contracts
(these are the most standard financial derivatives;
see, e.g., \citealt[Chap.~2]{Hull:2021}, and \citealt{Duffie:1989}).
Futures contracts is an old idea (see, e.g., \citealt{Schaede:1989})
that arose gradually in financial industry,
but in our prediction protocols it is a powerful way
of reducing prediction multiple steps ahead to one-step-ahead prediction.
In this appendix we will only need a highly idealized picture of them
(Sect.~\ref{subsec:ideal}),
but later (Appendix~\ref{app:real}) we will discuss their real-life counterparts.

We will also need another piece of notation:
$\mathbf{Y}^{m:n}$ stands
for the set of all sequences of elements of $\mathbf{Y}$ of length between $m$ and $n$ inclusive
(so that $\mathbf{Y}^{0:n}$ stands for the sequences of elements of $\mathbf{Y}$ of length at most $n$,
and $\mathbf{Y}^{1:n}$ stands for the non-empty sequences of elements of $\mathbf{Y}$ of length at most $n$).

\subsection{Ideal futures contracts}
\label{subsec:ideal}

In Sect.~\ref{sec:test} we extended the forecasting picture of Sect.~\ref{sec:basic}
by allowing Sceptic to bet against Forecaster.
Betting was described in terms of tickets,
which are known as forward contracts in finance.
In our diachronic picture, however,
we allowed trade in tickets (they were sold and bought),
which essentially turned them 
into what is known as futures contracts in finance.
In this appendix we will talk about futures contracts explicitly
using very convenient terminology developed in finance.
Our terminology, however, will be slightly adapted to our needs
(for example, the unit of time will be a step rather than, e.g., a day,
and the trader will be called Sceptic).

A futures contract $\Phi$ has an \emph{expiration step} $m$.
The contract is settled at the end of step $m$;
namely, its final price $F_m^+$ is announced by Reality.
In the middle of step $n\in\{1,\dots,m\}$,
the current price $F_n$ of $\Phi$ is announced by Forecaster,
and Sceptic can then take any \emph{position} $f_n\in\R$ in $\Phi$.
If $n<m$, Sceptic gains capital $f_n(F_{n+1}-F_n)$ at the next step
(which actually means losing capital if $f_n(F_{n+1}-F_n)<0$).
If $n=m$, at the end of the expiration step $m$ (at \emph{maturity})
Sceptic gains $f_{m}(F^+_m-F_{m})$.
These gains keep accumulating as the play proceeds.

\subsection{General testing protocol}

The following extension of Protocol~\ref{prot:joint} describes another,
seemingly stronger than Protocol~\ref{prot:joint-test-final},
way of testing Forecaster's predictions.
(However, we will reduce the extended protocol
to Protocol~\ref{prot:joint-test-final}
in Sect.~\ref{subsec:final}.)

\begin{protocol}\label{prot:joint-test-general}\ \\
  \indentI $\K_0 := 1$\\
  \indentI Forecaster announces $P_1\in\mathfrak{P}(\mathbf{Y}^{N})$\\
  \indentI Sceptic announces $f_1\in\R^{\mathbf{Y}^{1:N}}$\\
  \indentI Reality announces $y_1\in\mathbf{Y}$\\
  \indentI $\K'_1 := \K_0 + f_1(y_1) - \sum_{y\in\mathbf{Y}} f_1(y) P_1(y)$\\
  \indentI FOR $n=2,\dots,N$:\\
    \indentII Forecaster announces $P_n\in\mathfrak{P}(\mathbf{Y}^{N-n+1})$\\
    \indentII $\K_{n-1} := \K'_{n-1}
        + \sum_{x\in\mathbf{Y}^{1:(N-n+1)}} f_{n-1}(y_{n-1}x) P_{n}(x)$\\
        \indentIV ${}-\sum_{x\in\mathbf{Y}^{2:(N-n+2)}} f_{n-1}(x) P_{n-1}(x)$%
        \hfill\refstepcounter{equation}(\theequation)\label{eq:joint-K-P-0}\\
    \indentII Sceptic announces $f_n\in\R^{\mathbf{Y}^{1:(N-n+1)}}$\\
    \indentII Reality announces $y_n\in\mathbf{Y}$\\
    \indentII $\K'_n := \K_{n-1} + f_n(y_n) - \sum_{y\in\mathbf{Y}} f_n(y) P_n(y)$.
        \hfill\refstepcounter{equation}(\theequation)\label{eq:joint-K-y-0}
\end{protocol}

\noindent
Protocol~\ref{prot:joint-test-general} does not define $\K_N$,
and we set $\K_N:=\K'_N$.
The interpretation of $\K_N$ is the same as for Protocol~\ref{prot:joint-test-final}:
a large $\K_N$ evidences lack of agreement of the forecasts with reality,
provided Sceptic is not allowed to go into debt.

An advantage of Protocol~\ref{prot:joint-test-general} over Protocol~\ref{prot:joint-test-final}
is that, even though it is stated for a finite time horizon $N$,
it is easier to modify to make the time horizon infinite,
so that $n=2,3,\dots$ in the FOR loop.
The simpler protocol of Sect.~\ref{sec:test}
uses the finiteness of the time horizon in a more essential way.

The financial interpretation of Protocol~\ref{prot:joint-test-general}
is that we have a market of futures contracts $\Phi(x)$,
$x\in\mathbf{Y}^{1:N}$,
that pay
\[
  F_m^+(x):=1_{\{y_1\dots y_m = x\}}
\]
at the end of step $m:=\lvert x\rvert$,
as discussed in Sect.~\ref{subsec:ideal}.
At each step $n$ (but before observing $y_n$)
Forecaster announces the prices for all the futures contracts
$\Phi(x)$, $x\in\mathbf{Y}^{1:N}$,
in the form of a probability measure $P_n\in\mathfrak{P}(\mathbf{Y}^{N-n+1}$);
namely, the price of $\Phi(x)$, $x\in\mathbf{Y}^{1:N}$, at step $n$ is
\begin{equation}\label{eq:F-P}
  F_n(x)
  :=
  \begin{cases}
    P_n(x\setminus y_1\dots y_{n-1}) & \text{if $y_1\dots y_{n-1}\subset x$}\\
    0 & \text{if not}.
  \end{cases}
\end{equation}
A standard argument shows that such $P_n$ will exist
provided the market is coherent;
Sceptic can secure a sure gain if $F_n$ do not form a probability measure
concentrated on the continuations of $y_1\dots y_{n-1}$
\citep[Chap.~3]{deFinetti:2017}.

At step $n$ Sceptic needs to take positions
in all $\Phi(x)$, $y_1\dots y_{n-1}\subset x\in\mathbf{Y}^{1:N}$.
The position in $\Phi(y_1\dots y_{n-1}x)$ is denoted by $f_n(x)$
in Protocol~\ref{prot:joint-test-general}.
(There is no need to take positions in the other $\Phi(x)$
since their prices are 0 and will stay 0.)

After $y_n$ is disclosed by Reality,
the increment in Sceptic's capital
(due to the matured futures contracts $\Phi(y_1\dots y_{n-1}y)$)
is
\begin{align*}
  \K'_n-\K_{n-1}
  &=
  \sum_{y\in\mathbf{Y}}
  f_n(y)
  \left(
    F_n^+(y_1\dots y_{n-1}y)
    -
    F_n(y_1\dots y_{n-1}y)
  \right)\\
  &=
  f_n(y_n)
  -
  \sum_{y\in\mathbf{Y}}
  f_n(y)
  P_n(y),
\end{align*}
which agrees with \eqref{eq:joint-K-y-0}.
And after $P_n$ is disclosed by Forecaster at the next step $n:=n+1$,
the increment in Sceptic's capital
(due to the remaining futures contracts)
is
\begin{align*}
  \K_{n-1}-\K'_{n-1}
  &=
  \sum_{x\in\mathbf{Y}^{2:(N-n+2)}}
  f_{n-1}(x)
  \bigl(
    F_{n}(y_1\dots y_{n-2} x)
    -
    F_{n-1}(y_1\dots y_{n-2} x)
  \bigr)\\
  &=
  \sum_{x\in\mathbf{Y}^{1:(N-n+1)}} f_{n-1}(y_{n-1}x) F_{n}(y_1\dots y_{n-1}x)\\
  &\quad-
  \sum_{x\in\mathbf{Y}^{2:(N-n+2)}} f_{n-1}(x) F_{n-1}(y_1\dots y_{n-2}x)\\
  &=
  \sum_{x\in\mathbf{Y}^{1:(N-n+1)}} f_{n-1}(y_{n-1}x) P_{n}(x)
  -
  \sum_{x\in\mathbf{Y}^{2:(N-n+2)}} f_{n-1}(x) P_{n-1}(x),
\end{align*}
where the last equality follows from \eqref{eq:F-P},
and the last expression agrees with \eqref{eq:joint-K-P-0}.

\subsection{Simplification}
\label{subsec:final}

We can rewrite Protocol~\ref{prot:joint-test-general} in other forms,
such as Protocol~\ref{prot:joint-test-final},
getting rid of some of Sceptic's arbitrary choices.
To compare protocols with the same allowed moves for Reality and Forecaster,
we can use the notion of the \emph{test martingale space} (TMS),
which we define modifying the definition given in Sect.~\ref{sec:embedding}
as follows.
A strategy for Sceptic still specifies his move
as function of Forecaster's and Reality's previous moves,
but now we do not impose any measurability conditions on strategies.
As before, the corresponding \emph{test martingale} is Sceptic's capital
as function of Forecaster's and Reality's moves provided this function is nonnegative.
The TMS of a given protocol is the set of all possible test martingales.
We regard two protocols to be equivalent if they have the same TMS.

As already mentioned,
the general testing protocol, Protocol~\ref{prot:joint-test-general},
was formulated with a view towards an infinite time horizon,
where $N$ becomes $\infty$.
In Sect.~\ref{sec:test} we introduced a much simpler protocol
using an idea that only works for a finite time horizon.

\begin{proposition}\label{prop:equivalent-final}
  Protocol~\ref{prot:joint-test-general} and Protocol~\ref{prot:joint-test-final}
  have identical TMS.
\end{proposition}

\noindent
Proposition~\ref{prop:equivalent-final}
simplifies the market in futures contracts that we need:
all the contracts now mature at the end of step $N$;
we will call such futures contracts \emph{final}.
The intuitive reason why the final futures contracts are sufficient
is that a general futures contract $\Phi(x)$ is equivalent,
to all intents and purposes,
to the portfolio consisting of the final futures contracts $\Phi(x')$
for all $x'\supseteq x$.

\begin{proof}[Proof of Proposition~\ref{prop:equivalent-final}]
  Consider step $n<N$ of Protocol~\ref{prot:joint-test-general}.
  Let $O(x,c)$, where $x\in\mathbf{Y}^{1:(N-n)}$ and $c\in\R$,
  be the operation that adds the constant $c$ to $f_n(x)$
  and subtracts the same constant $c$ from all $f_n(x y)$, $y\in\mathbf{Y}$.
  The key observation used in our simplification of Protocol~\ref{prot:joint-test-general}
  is that, for any $x\in\mathbf{Y}^{1:(N-n)}$ and $c\in\R$,
  $O(x,c)$ does not change the increment in the capital $\K_n-\K_{n-1}$.
  Let us check this property.
  If $x\in\mathbf{Y}^{2:(N-n)}$,
  $O(x,c)$ will not affect \eqref{eq:joint-K-y-0} whatsoever,
  and it will change neither minuend nor subtrahend in \eqref{eq:joint-K-P-0}
  at the next step (there is a next step since $n<N$).
  And if $x\in\mathbf{Y}$,
  applying the operation $O(x,c)$ does not change
  the increment in the capital $\K_n-\K_{n-1}$
  given by \eqref{eq:joint-K-y-0} and then by \eqref{eq:joint-K-P-0} at the next step
  since
  \begin{itemize}
  \item
    the changes in the sum in \eqref{eq:joint-K-y-0}
    and in the second sum in \eqref{eq:joint-K-P-0} at the next step
    will balance each other out, and
  \item
    the changes in the term $f_n(y_n)$ in \eqref{eq:joint-K-y-0}
    and in the first sum in \eqref{eq:joint-K-P-0} at the next step
    will also balance each other out
    (this is relevant only when $x=y_n$).
  \end{itemize}

  Applying $O(x,c)$ repeatedly to the $x$s in the order of increasing length,
  we can assume, without loss of generality
  (i.e., without changing the TMS),
  that $f_n(x)$ is different from $0$ only for $x\in\mathbf{Y}^{N-n+1}$,
  which implies that:
  \begin{itemize}
  \item
    we can ignore \eqref{eq:joint-K-y-0} for all steps $n$ apart from $n=N$,
    and so \eqref{eq:joint-K-y-2} is performed only for $n=N$;
  \item
    we can ignore the bits ``$1:{}$'' and ``$2:{}$'' in \eqref{eq:joint-K-P-0},
    obtaining \eqref{eq:joint-K-P-2}.
  \end{itemize}
  Protocol~\ref{prot:joint-test-final} also merges
  the four lines in Protocol~\ref{prot:joint-test-general} preceding the FOR loop
  into the loop.
\end{proof}

\section{Mechanics of futures trading}
\label{app:real}

Section~\ref{subsec:ideal} gives an idealized picture of futures trading.
The main elements of simplification in it are:
\begin{itemize}
\item
  the interest rate is assumed to be zero;
\item
  the positions and futures prices are assumed to take any real values
  (although we are only interested in positive prices for futures contracts);
\item
  there is no difference between the selling and buying prices
  (no bid/ask spread);
\item
  there are no other transaction costs.
\end{itemize}

In this paper we are only interested in binary futures contracts
(where the outcome is 0 or 1).
However, the most popular market mechanism,
described in this appendix,
works for general futures contracts,
which are not restricted to the binary case.

A good reference for traditional futures markets is \citet{Duffie:1989}.
While some of the physical details of trading described in it might be obsolete,
the general principles are still applicable.
Another good reference is \citet{Harris:2003}.

By far the most popular platform for prediction markets
is the Iowa Electronic Markets (IEM).
The IEM was created in 1988 and has always been a small-scale operation;
the development of prediction markets has been greatly hindered
by the US anti-gambling regulation \citep{Arrow/etal:2008}.
The IEM was created by academics, and its role is mainly educational;
in particular, it has a great help system explaining the market microstructure
(which I often follow in this section).
It received two no-action letters, in 1992 and 1993,
from the US Commodity Futures Trading Commission (CFTC)
reducing the chance of legal action against it.
Its competitors sometimes have better bid/ask spreads,
but their positions are less secure;
e.g., Intrade (1999--2013) is now defunct
and PredictIt (launched in 2014) had their CFTC no-action letter withdrawn in 2022.

A futures contract is a contract that pays a specified amount $F_m$
at a specified future time, called the \emph{expiration time} $m$
(it was the expiration step in the main part of the paper).
The amount is uncertain at the time of trading but becomes well-defined
at the expiration time, when trading ceases.
An example of $m$ and $F_m$ is ``6 November 2024''
and ``Democratic Nominee's share of the two-party popular vote
in the 2024 US Presidential election'' in US dollars.
This is, essentially,
one of the types of futures contracts traded at the IEM in August 2023
(of the ``vote share'' variety; the other main variety is ``winner takes all'').
Let us fix $m$ and $F_m$.
At each time the market participants can hold any number of the futures contracts
(positive, zero, or negative),
which is known as their \emph{positions} in the futures contracts.
They can also submit orders to change their positions.
The main kinds of orders are \emph{market orders} and \emph{limit orders}.
A limit order specifies the number of futures contracts to buy or sell
at a given price (known as the \emph{bid price} for orders to buy
and the \emph{ask price} for orders to sell);
it may also specify the time when the order expires.

At the core of a futures market is the \emph{order book}
listing the outstanding limit orders.
The prices specified in those orders are
\begin{equation}\label{eq:prices}
  B_{n_B}<B_{n_B-1}<\dots<B_1<A_1<A_2<\dots<A_{n_A},
\end{equation}
where $n_B$ is the number of different bid prices in the currently active limit orders to buy
and $n_A$ is the number of different ask prices in the currently active limit orders to sell.
The prices in the list \eqref{eq:prices} are sorted in the ascending order,
and the difference $A_1-B_1$ is known as the \emph{bid/ask spread}.
With each price level $x$ is associated the total number $N(x)$ of futures contracts
that the market participants with active limit orders wish to trade
(to buy if $x=B_n$ for some $n$ and to sell if $x=A_n$ for some $n$;
$N(x)=0$ for all other $x$).
The order book consists of the prices \eqref{eq:prices}
and the numbers $N(x)$ of futures contracts offered at each price level $x$
(within each price level $x$ older orders appear before newer orders).
It consists of a \emph{bid queue}
(the data related to the bid prices)
and an \emph{ask queue}
(the data related to the ask prices).

A market order is simpler than a limit order
and only specifies the number of futures contracts to buy or sell.
When a new market order is submitted by a market participant MP,
it is matched with the order book immediately and a trade is performed.
Namely, if the order is to sell $N$ contracts,
the bid queue is traversed from the top (i.e., from $B_1$)
until the required number of orders to buy is found:
we find the smallest $k$ such that $N(B_1)+\dots+N(B_k)\ge N$
(all the $N(B_1)+\dots+N(B_{n_B})$ contracts requested in the bid queue
are bought if there is no such $k$)
and arrange a trade with MP selling all his futures contracts
to the market participants with active limit orders with the prices in $\{B_1,\dots,B_k\}$;
for the price $B_k$ only the oldest orders are fulfilled (perhaps partially).
The procedure for market orders to buy is analogous.

When a new limit order is submitted by a market participant,
it is simply added to the order book.
We can assume that the limit orders to buy specify prices below $A_1$
and the limit orders to sell specify prices above $B_1$
(otherwise, a market order can be submitted).
When a limit order in the order book expires,
it is, of course, removed from it.

An important element of futures markets is the system of \emph{margins}.
Typically market participants have positions in several types of futures contracts
(corresponding to different $m$ and $F_m$)
and other securities,
and the total values of their portfolios can go up or down.
To reduce the chance of the exchange losing money,
they are required to maintain margin accounts at specified levels.
If a margin account falls below the specified level
as result of changing market prices,
a \emph{margin call} is issued requiring the account to be replenished.

In the IEM, short (i.e., negative) positions are formally prohibited,
which allows it to avoid imposing margin requirements.
But it is still easy to emulate short positions
(e.g., a short position in the vote share for the Democratic Nominee
can be modelled as a long position in the vote share for the Republican Nominee).

A natural question is how a futures market is started;
namely how to make the order book non-empty.
In the IEM, the market participants are allowed to buy \emph{fixed price bundles}
for a given price.
For example, such a bundle might contain the vote share for the Democratic Nominee
and the vote share for the Republican Nominee,
with a fixed price of $\$1$
(the sum of the two vote shares is 1,
and so the final pay-off of the bundle is known to be $\$1$).

\section{Radical probabilism}
\label{app:radical}

Our testing protocols,
such as Protocol~\ref{prot:joint-test-final},
assume that we learn the observations $y_n$ with full certainty.
According to Jeffrey's doctrine of radical probabilism \citep{Jeffrey:1992},
we do not learn anything for certain;
at best, we learn that the $n$th observation is $y_n$
with a high probability.
The uncertainty of observations
is a recurring topic in the philosophy of science.
See, e.g., Popper's discussion of ``basic statements''
in \citet[Chap.~5]{Popper:1959}
(where he also refers to Reininger's and Neurath's similar ideas)
and \citet{Andersson:2016}.
In this section we will discuss two modifications
of Protocol~\ref{prot:joint-test-final}
allowing uncertain evidence.

\subsection{Additive picture}
\label{subsec:radical-additive}

A straightforward modification of Protocol~\ref{prot:merged}
making evidence uncertain is the following one.

\begin{protocol}\label{prot:radical-additive}\ \\
  \indentI $\K_0 := 1$\\
  \indentI FOR $n=1,2,\dots$:\\
    \indentII Forecaster announces $Q_n\in\mathfrak{P}(\mathbf{Y}^{N})$\\
    \indentII IF $n>1$:\\
    \indentIII $\K_{n-1} := \K_{n-2}
        + \sum_{x\in\mathbf{Y}^{N}} F_{n-1}(x) (Q_{n}(x) - Q_{n-1}(x))$
        \text{\quad}\hfill\refstepcounter{equation}(\theequation)\label{eq:radical-add}\\
    \indentII Sceptic announces $F_n\in\R^{\mathbf{Y}^{N}}$.
\end{protocol}

\noindent
Whereas the loops in Protocols~\ref{prot:joint-test-final} and~\ref{prot:merged}
are over finite ranges of $n$,
in Protocol~\ref{prot:radical-additive} the loop is infinite
since we do not learn any of $y_1,\dots,y_N$ with certainty.
Even though in Protocol~\ref{prot:radical-additive} $y_n$ are never disclosed explicitly,
they may be disclosed implicitly via $Q_n$: cf.\ \eqref{eq:Q}.
The capital updating rule \eqref{eq:radical-add} is very natural:
namely, a possible interpretation of this rule
is that $Q_{n-1}$ is the expectation of $Q_{n}$
(cf.\ \citealt[Theorem in Sect.~3]{Goldstein:1983}).

Notice that Protocol~\ref{prot:radical-additive}
is only a modification, not  generalization, of Protocol~\ref{prot:merged}:
whereas $Q_n$ in the former protocol is required to be positive,
it is not positive in the latter protocol
(being positive is incompatible with possessing certain evidence).

\subsection{Multiplicative picture}

Protocol~\ref{prot:radical-additive}
and all the protocols discussed in the main part of the paper
are similar to Protocol~\ref{prot:joint-test-final}
in that Sceptic's capital is updated by adding various terms.
This subsection introduces a multiplicative protocol,
in which Sceptic's capital is updated by multiplication.
Both multiplicative and additive protocols
are ubiquitous in game-theoretic probability
(although the difference between them is rarely pointed out).
This is the multiplicative version of Protocol~\ref{prot:radical-additive}:

\begin{protocol}\label{prot:radical-multiplicative}\ \\
  \indentI $\K_0 := 1$\\
  \indentI FOR $n=1,2,\dots$:\\
    \indentII Forecaster announces $Q_n\in\mathfrak{P}(\mathbf{Y}^{N})$\\
    \indentII IF $n>1$:\\
    \indentIII $\K_{n-1} := \K_{n-2}
      \sum_{x\in\mathbf{Y}^{N}} \frac{Q_{n}(x)}{Q_{n-1}(x)} G_{n-1}(x)$
      \hfill\refstepcounter{equation}(\theequation)\label{eq:radical-multi}\\
    \indentII Sceptic announces $G_n\in\R^{\mathbf{Y}^{N}}$.
\end{protocol}

As usual, $\K_n$ is not allowed to become negative.
To see the equivalence of the additive and multiplicative protocols,
notice that \eqref{eq:radical-multi} is equivalent to
\begin{align*}
  \K_{n-1} - \K_{n-2}
  &=
  \left(
    \sum_{x\in\mathbf{Y}^{N}}
    \frac{Q_{n}(x)}{Q_{n-1}(x)} G_{n-1}(x)
    -
    1
  \right)
  \K_{n-2}\\
  &=
  \left(
    \sum_{x\in\mathbf{Y}^{N}}
    (Q_{n}(x)-Q_{n-1}(x)) \frac{G_{n-1}(x)}{Q_{n-1}(x)}
  \right)
  \K_{n-2}.
\end{align*}
This establishes the one-to-one correspondence
\begin{equation}\label{eq:correspondence}
  F_{n-1}(x)
  =
  \frac{G_{n-1}(x)}{Q_{n-1}(x)}
  \K_{n-2}
\end{equation}
between $F_{n-1}$ in \eqref{eq:radical-add} and $G_{n-1}$ in \eqref{eq:radical-multi}.
The correspondence~\eqref{eq:correspondence} assumes that $\K_{n-2}>0$,
and the case $\K_{n-2}=0$ should be considered separately
(Sceptic' capital will stay at 0 once it reaches 0 in either protocol).

We obtain a useful modification of Protocol~\ref{prot:radical-multiplicative}
replacing $G_n\in\R^{\mathbf{Y}^{N}}$ in the last line
by $G_n\in\mathfrak{P}(\mathbf{Y}^{N})$.
Then the multiplicative protocol becomes a special case of Cover's protocol
modelling investment into $\lvert\mathbf{Y}\rvert^N$ securities
such as stocks (see, e.g., \citealt{Cover:1991} or \citealt[Example 9]{Vovk:1998}).
As in Sect.~\ref{sec:test},
we have a market in securities $\Phi(x)$, $x\in\mathbf{Y}^N$,
but they may be never settled.
For each security $\Phi(x)$ the protocol gives
its price $Q_n(x)$ at time $n$.
The prices are normalized in that $Q_n(x)$ sum to $1$ over $x$;
e.g., $Q_n(x)$ may be the market shares.
The capital update rule \eqref{eq:radical-multi}
involves the \emph{price relative} $Q_{n}(x)/Q_{n-1}(x)$
(as used in \citealt{Cover:1991}).
At each step Sceptic decides on the distribution $G_n$
of his current capital $\K_{n-1}$
among the securities $\Phi(x)$.
If $G_n\in\mathfrak{P}(\mathbf{Y}^{N})$,
we do not allow ``short selling'',
i.e., holding a negative amount of a security,
and we require Sceptic to invest all of his capital.
In general, allowing any $G_n\in\R^{\mathbf{Y}^{N}}$
we allow both short selling
and leaving part (positive or negative) of Sceptic's capital
on a zero-interest bank account.

\subsection{Radical probabilism and reality}

The additive picture and, especially, the multiplicative one
shed new light on the protocols in the main part of the paper.
The latter cover the case where $Q_n$, $n=1,\dots,N$,
is concentrated on $[x]\subseteq\mathbf{Y}^N$
(the set of all continuations of $x$)
for some $x\in\mathbf{Y}^{N-n+1}$.
The difference between radical probabilism
and the standard Bayesian scenario considered in the main paper
corresponds to the difference between stocks and futures contracts.
Sooner or later, reality settles a futures contract,
but stock prices can be forever variable (in our ideal picture).

It would be interesting to establish conditions
under which this paper's results can be extended
to the more general and simpler protocols of this appendix.

\section{Measure-theoretic martingale law of large numbers}
\label{app:martingale-SLLN}

Our discussion of Bayesian decision theory in Sect.~\ref{sec:DM}
was based on a law of large numbers for predicting $K$ steps ahead.
This law of large numbers may also present an independent interest,
and the purpose of this appendix is to give
clean self-contained measure-theoretic statements of its various versions.
In this appendix we consider general probability spaces $(\Omega,\FFF,P)$,
not necessarily finite.

A \emph{filtration} $(\FFF_n)$, $n=0,1,\dots,N$, in a general probability space $(\Omega,\FFF,P)$
is still an increasing sequence of $\sigma$-algebras, $\FFF_0\subseteq\dots\subseteq\FFF_N$.
A sequence $Y_1,\dots,Y_N$ of random variables in $(\Omega,\FFF,P)$ is \emph{adapted}
if $Y_n$ is $\FFF_n$-measurable for $n=1,\dots,N$.
We usually assume $\lvert Y_n\rvert\le1$ for agreement with the assumption $\lambda_n\in[0,1]$
that we made in Sect.~\ref{sec:DM} about the loss functions:
$Y_n$ corresponds to a difference between two values of such a loss function $\lambda_n$.

Interestingly, we can get nearly optimal results by using the primitive idea
of decomposing forecasting $K$ steps ahead into $K$ processes of forecasting one step ahead,
as in Remark~\ref{rem:strong-law}.
This gives us the following proposition (analogous to Theorem~\ref{thm:optimal}).

\begin{proposition}\label{prop:LLN}
  Let $(\Omega,\FFF,P)$ be a probability space
  equipped with a filtration $(\FFF_n)$, $n=0,1,\dots,N$.
  Fix a prediction horizon $K\in\{1,\dots,N\}$.
  Let $Y_1,\dots,Y_N$ be an adapted sequence of random variables
  in $(\Omega,\FFF,P)$
  bounded by 1 in absolute value, $\lvert Y_n\rvert\le1$ for $n=1,\dots,N$.
  Then we have, for any $\epsilon\in(0,0.7)$,
  \begin{equation}\label{eq:prop-LLN}
    P
    \left(
      \left|
        \sum_{n=K}^N
        \left(
          Y_n - \E_{P}(Y_n\mid\FFF_{n-K})
        \right)
      \right|
      \ge
      4 \sqrt{K N\ln\frac{1}{\epsilon}}
    \right)
    \le
    \epsilon.
  \end{equation}
\end{proposition}

\begin{proof}
  In this proof we will need one result from robust risk aggregation
  (this theory was originated by Kolmogorov \citep{Makarov:1981};
  it is briefly described in \citealt[Remark~2]{Vovk/Wang:2020}
  and then widely used in that paper).
  Namely, we will need the following special case of Theorem~4.2
  of \citet{Embrechts/Puccetti:2006}.

  Suppose nonnegative random variables $X_k$, $k=1,\dots,K$,
  satisfy
  \begin{equation}\label{eq:marginals}
    \P(X_k\ge x)
    =
    \exp(-a x^2)
  \end{equation}
  for all $x\ge0$,
  where $a$ is a positive constant.
  The value $E$ of the optimization problem
  \begin{equation}\label{eq:optimization}
    \P(X_1+\dots+X_K\ge C)
    \to
    \max
  \end{equation}
  (the $\max$, or at least $\sup$,
  being over all joint distributions for $(X_1,\dots,X_K)$
  with the given marginals)
  does not exceed
  \begin{equation}\label{eq:E}
    E
    :=
    \inf_{t<C/K}
    \frac{K\int_t^{C-(K-1)t}\exp(-a x^2)\d x}{C-K t}.
  \end{equation}

  We can extend the statement in the previous paragraph
  to a wider class of random variables $X_k$, $k=1,\dots,K$.
  Namely, it suffices to assume that they satisfy
  \begin{equation}\label{eq:relaxed-marginals}
    \P(X_k\ge x)
    \le
    \exp(-a x^2)
  \end{equation}
  for all $x\ge0$, instead of \eqref{eq:marginals}.
  We will apply the statement to the random variables $X_k$ given by
  \begin{equation*}
    X_k
    :=
    \sum_{n\in\{k+K,k+2K,\dots,k+\lfloor N/K\rfloor K\}}
    \left(
      Y_{n} - \E_{P}(Y_{n}\mid\FFF_{n-K})
    \right).
  \end{equation*}
  By Hoeffding's inequality,
  for any $C>0$ and any $k\in\{0,\dots,K-1\}$,
  \begin{equation*}
    P
    \left(
      X_k
      \ge
      C
    \right)\\
    \le
    \exp
    \left(
      -C^2/(2\lfloor N/K\rfloor)
    \right)
    \le
    \exp
    \left(
      -C^2/(2N/K)
    \right),
  \end{equation*}
  where the non-existent terms in the sum
  (those corresponding to $n>N$ if any)
  are interpreted as 0.
  Therefore, \eqref{eq:relaxed-marginals} holds with
  \begin{equation}\label{eq:a}
    a
    :=
    \frac{K}{2N}.
  \end{equation}

  Let us set $t:=\frac{C}{2K}$ in \eqref{eq:E}
  (this is the middle of the range of $t$).
  This gives the upper bound
  \[
    \frac{2K}{C}
    \int_{\frac{C}{2K}}^{\infty}
    \exp(-a x^2)
    \d x
  \]
  for $E$,
  which can be rewritten (see below for an explanation) as
  \begin{align}
    &\frac{2K}{C}
    \frac{1}{\sqrt{2a}}
    \int_{\sqrt{2a}\frac{C}{2K}}^{\infty}
    \exp(-y^2/2)
    \d y
    =
    \frac{2K}{C}
    \frac{\sqrt{2\pi}}{\sqrt{2a}}
    \bar\Phi
    \left(
      \sqrt{2a}\frac{C}{2K}
    \right)\label{eq:notation}\\
    &=
    \frac{2\sqrt{2\pi}\sqrt{K N}}{C}
    \bar\Phi
    \left(
      \frac{C}{2\sqrt{K N}}
    \right)
    <
    \frac{4 K N}{C^2}
    \exp
    \left(
      -\frac{C^2}{8 K N}
    \right)\label{eq:Feller}.
  \end{align}
  The first expression in~\eqref{eq:notation}
  is obtained by the substitution $y:=\sqrt{2a}x$,
  the equality in \eqref{eq:notation} uses the notation $\bar\Phi$
  for the survival function of the standard Gaussian distribution,
  the following equality (the one in~\eqref{eq:Feller})
  is obtained by plugging in \eqref{eq:a},
  and the final inequality in~\eqref{eq:Feller}
  follows from the usual upper bound for $\bar\Phi$
  \citep[Lemma~VII.1.2]{Feller:1968}.

  To find a suitable solution to the inequality
  \[
    \frac{4 K N}{C^2}
    \exp
    \left(
      -\frac{C^2}{8 K N}
    \right)
    \le
    \frac{\epsilon}{2},
  \]
  we plug in $C=\sqrt{8 K N \ln\frac{1}{\epsilon} x}$
  (intuitively, $x\approx1$)
  obtaining, after simplification,
  \[
    \epsilon^{x-1}
    \le
    x \ln\frac{1}{\epsilon}.
  \]
  Assuming $\epsilon<0.7$,
  we can set $x:=2$.
\end{proof}

\begin{remark}
  In the proof of Proposition~\ref{prop:LLN}
  we did not make any attempt to optimize
  the coefficient $4$ in \eqref{eq:prop-LLN}.
  However, the same argument shows that 4
  can be replaced by a number as close to $\sqrt{2}$ as we wish
  if we narrow down the permitted range of $\epsilon$
  (leaving the lower end of the range at 0, of course).
\end{remark}

\begin{remark}
  Since the bound $E$ in \eqref{eq:E} plays an important role
  in this appendix (and implicitly in Appendix \ref{subsec:optimal}),
  it is reassuring to know that in many interesting cases
  $E$ actually coincides with the value of the optimization problem~\eqref{eq:optimization}.
  This is shown in Theorem~2.3 by \citet{Puccetti/Ruschendorf:2013}.
  (And the restatement of Embrechts and Puccetti's result
  in \citealt[Sect.~1]{Puccetti/Ruschendorf:2013},
  is particularly convenient.)
  One of the cases \citep[Sect.~3]{Puccetti/Ruschendorf:2013}
  in which $E$ is the value of the optimization problem
  is where the probability density function of $X_k$
  is monotonically decreasing over its domain $[0,\infty)$.
  This condition, however, is only satisfied
  for $x\ge1/\sqrt{2a}$
  (the last condition becomes $x\ge\sqrt{N/K}$
  for the value of $a$, given in \eqref{eq:a},
  that we will be interested in).
\end{remark}

\begin{remark} 
  In the proof of Proposition~\ref{prop:LLN}
  we set $t:=\frac{C}{2K}$ in \eqref{eq:E}.
  In the arXiv version~2 of this paper,
  I used two other choices,
  $t\to\frac{C}{K}$ and $t:=0$,
  which led to weaker results
  (if we ignore the coefficient in front of the $\surd$ in \eqref{eq:prop-LLN}).
  Namely,
  the former choice is equivalent to using Bonferroni's inequality
  (as noticed by Embrechts and Puccetti
  \citep[Remark 4.1(i)]{Puccetti/Ruschendorf:2013}),
  and the latter choice gives a worse dependence of $\epsilon$,
  namely $\epsilon^{-2}$ in place of $\ln\frac1\epsilon$.
\end{remark}

Let us state Proposition~\ref{prop:LLN} in a cruder way.
Now we consider a sequence of probability spaces $(\Omega_N,\FFF_N,P_N)$,
$N=1,2,\dots$, each equipped with a filtration $(\FFF_{N,n})$, $n=0,1,\dots,N$.
Fix a sequence $K_N\in\{1,\dots,N\}$, $N=1,2,\dots$, of prediction horizons.
Let, for each $N$, $Y_{N,1},\dots,Y_{N,N}$ be an adapted sequence of random variables
in $(\Omega_N,\FFF_N,P_N)$
bounded by 1 in absolute value, $\lvert Y_{N,n}\rvert\le1$ for $n=1,\dots,N$.
When I say that a relation $R_N(O(X_N))$ involving $O(X_N)$
(such as \eqref{eq:LLN-simple} below)
holds in probability,
I mean that for any $\epsilon>0$ there exists $C>0$
such that $P_N(R_N(C X_N))\ge1-\epsilon$ from some $N$ on.%
\footnote{Of course, this definition makes an intuitive sense
  only when the statement $R_N(x)$ becomes weaker as $x$ increases.}
According to~\eqref{eq:prop-LLN},
\begin{equation}\label{eq:LLN-simple}
  \left|
    \sum_{n=K_N}^N
    \left(
      Y_{N,n} - \E_{P_N}(Y_{N,n}\mid\FFF_{N,n-K_N})
    \right)
  \right|
  =
  O
  \left(
    \sqrt{K_N N}
  \right)
\end{equation}
in probability.
An even cruder form of \eqref{eq:LLN-simple}
(and of Proposition~\ref{prop:LLN})
is the following corollary.
\begin{corollary}
  Let $(\Omega_N,\FFF_N,P_N)$, $N=1,2,\dots$,
  be a sequence of probability spaces $(\Omega_N,\FFF_N,P_N)$
  each equipped with a filtration $(\FFF_{N,n})$, $n=0,1,\dots,N$.
  Suppose the sequence $K_N\in\{1,\dots,N\}$, $N=1,2,\dots$, of prediction horizons
  satisfies $K_N=o(N)$.
  Let, for each $N$, $Y_{N,1},\dots,Y_{N,N}$ be an adapted sequence of random variables
  in $(\Omega_N,\FFF_N,P_N)$
  bounded by 1 in absolute value.
  Then
  \begin{equation}\label{eq:cor-LLN}
    \left|
      \frac{1}{N-K_N+1}
      \sum_{n=K_N}^N
      \left(
        Y_{N,n} - \E_{P_N}(Y_{N,n}\mid\FFF_{N,n-K_N})
      \right)
    \right|
    \to
    0
    \quad
    (N\to\infty)
  \end{equation}
  holds in probability.
\end{corollary}

\noindent
Remember that when we say that random variables $\xi_N$
in probability spaces $(\Omega_N,\FFF_N,P_N)$
converge to 0 in probability,
as in \eqref{eq:cor-LLN},
we mean that, for any $\delta>0$,
$P_N(\lvert\xi_N\rvert>\delta)\to0$ as $N\to\infty$.

The following proposition
(analogous to Proposition~\ref{prop:lower-bound-1})
is an inverse to \eqref{eq:LLN-simple}.
To make it slightly stronger,
we state it for finite probability spaces.
\begin{proposition}\label{prop:anti-LLN}
  There exist $\epsilon>0$,
  a sequence of finite probability spaces $(\Omega_N,P_N)$, $N=1,2,\dots$,
  each equipped with a filtration $(\FFF_{N,n})$, $n=0,1,\dots,N$,
  and, for each $N$,
  an adapted sequence $Y_{N,1},\dots,Y_{N,N}$ of random variables in $(\Omega_N,P_N)$
  bounded by 1 in absolute values, $\lvert Y_{N,n}\rvert\le1$ for $n=1,\dots,N$,
  such that, for any sequence $K_N\in\{1,\dots,\lfloor N/5\rfloor\}$, $N=5,6,\dots$,
  and for all $N\ge5$,
  we have
  \begin{equation*}
    P_N
    \left(
      \sum_{n=K_N}^N
      \left(
        Y_{N,n} - \E_{P_N}(Y_{N,n}\mid\FFF_{N,n-K_N})
      \right)
      \ge
      \sqrt{K_N N}
    \right)
    \ge
    \epsilon.
  \end{equation*}
\end{proposition}

\begin{proof}
  Fix independent $\{-1,1\}$-valued variables
  $X_1,\dots,X_{\lceil N/K_N\rceil}$ in $(\Omega_N,P_N)$
  taking values $\pm1$ with equal probabilities,
  and set
  \begin{equation*}
    Y_{N,n}
    :=
    X_{\lceil n/K_N\rceil},
    \quad
    n=1,\dots,N.
  \end{equation*}
  Therefore, the $N$ steps are split into $\lceil N/K_N\rceil$ blocks of length $K_N$
  (with a possible exception of the last block, which may be shorter),
  and $Y_{N,n}$ is constant within each block.
  By the central limit theorem,
  the probability is at least $\epsilon$ (a universal positive constant)
  that $Y_{N,n}=1$ in at least $\sqrt{N/K_N}+1$ more blocks than $Y_{N,n}=-1$.
  In such cases
  \[
    \sum_{n=K_N}^N
    \left(
      Y_{N,n} - \E_{P_N}(Y_{N,n}\mid\FFF_{N,n-K_N})
    \right)
    =
    \sum_{n=K_N}^N Y_{N,n}
    \ge
    K_N \sqrt{N/K_N}
    =
    \sqrt{K_N N},
  \]
  where each $\FFF_{N,n}$ is generated by $Y_{N,1},\dots,Y_{N,n}$.
\end{proof}

\begin{remark}\label{rem:inefficient}
  One inefficient approach to the $K$-steps ahead martingale law of large numbers
  (used in the arXiv version~1 of this paper
  and already alluded to in Remark~\ref{rem:strong-law})
  is to apply Hoeffding's inequality 
  to the martingale difference
  \[
    X_n
    :=
    \sum_{i=n}^{(n+K_N-1)\wedge N}
    \left(
      \E_{P_N}(Y_{N,i}\mid\FFF_{N,n})
      -
      \E_{P_N}(Y_{N,i}\mid\FFF_{N,n-1})
    \right),
  \]
  whose increments are bounded by $2K_N$ in absolute value.
  It is a martingale difference in the sense $\E(X_{n}\mid\FFF_{N,n-1})=0$, $n=1,\dots,N$,
  and it satisfies
  \begin{align*}
    \sum_{n=1}^N
    X_n
    &=
    \sum_{n=K_N}^{N}
    \left(
      Y_{N,n}
      -
      \E_{P_N}(Y_{N,n}\mid\FFF_{N,n-K_N})
    \right)\\
    &\qquad{}+
    \sum_{n=N+1}^{N+K_N-1}
    \left(
      \E_{P_N}(Y_{N,n}\mid\FFF_{N,N})
      -
      \E_{P_N}(Y_{N,n}\mid\FFF_{N,n-K_N})
    \right)\\
    &\approx
    \sum_{n=K_N}^{N}
    \left(
      Y_{N,n}
      -
      \E_{P_N}(Y_{N,n}\mid\FFF_{N,n-K_N})
    \right)
  \end{align*}
  (where the $\approx$ assumes $K_N\ll N$ and ignores borderline effects).
  This argument, however, requires $K_N=o(N^{1/2})$.
\end{remark}
\end{document}